\documentclass[preprint]{elsarticle}

\usepackage{array,xspace,multirow,hhline,tikz,colortbl,tabularx,booktabs,fixltx2e,amsmath,amssymb,amsfonts,amsthm}
\usepackage{algorithm}
\usepackage{algorithmic}
\usepackage{verbatim,ifthen}
\usepackage{enumitem}
\usepackage{pifont}
\usepackage{calrsfs,mathrsfs}
\usepackage{lscape}
\usetikzlibrary{arrows}
\usetikzlibrary{positioning}

	\newcommand{\eg}{e.g.,\xspace}

	\newtheorem{lemma}{Lemma}%
	\newtheorem{theorem}{Theorem}%
	\newtheorem{corollary}{Corollary}%
	\newtheorem{example}{Example}

			\newtheorem{remark}{Remark}

\definecolor{light-gray}{gray}{0.95}
	\newcommand{\pref}{\succsim \xspace}

	\newcommand{\nat}{\mathbb{N}}
	\newcommand{\natz}{\nat_0}
	\newcommand{\ceil}[1]{\left\lceil #1 \right\rceil}
	\newcommand{\floor}[1]{\left\lfloor #1 \right\rfloor}
	\newcommand{\card}[1]{\left| #1 \right|}

	\usepackage{enumitem}
	\setenumerate[1]{label=\rm(\it{\roman{*}}\rm),ref=({\it\roman{*}}),leftmargin=*}

	\newlength{\wordlength}

	\newcommand{\set}[1]{\{#1\}}
	\newcommand{\midd}{\mathbin{:}}

\usepackage{lineno}

\begin{document}

	\title{Fair Assignment of Indivisible Objects\\ under Ordinal Preferences}

		\author{Haris Aziz}\ead{haris.aziz@nicta.com.au}
	\author{Serge Gaspers} \ead{sergeg@cse.unsw.edu.au}
		\author{Simon Mackenzie} \ead{simon.mackenzie@nicta.com.au}
	\author{Toby Walsh} \ead{toby.walsh@nicta.com.au}

\address{NICTA and University of New South Wales, Sydney 2052, Australia} 


	\begin{abstract}
		We consider the discrete assignment problem in which agents express ordinal preferences over objects and these objects are allocated to the agents in a fair manner. We use the stochastic dominance relation between fractional or randomized allocations to systematically define varying notions of proportionality and envy-freeness for discrete assignments.  The computational complexity of checking whether a fair assignment exists is studied for these fairness notions.
		We also characterize the conditions under which a fair assignment is guaranteed to exist. For a number of fairness concepts, polynomial-time algorithms are presented to check whether a fair assignment exists. Our algorithmic results also extend to the case of unequal entitlements of agents. Our NP-hardness result, which holds for several variants of envy-freeness, answers an open question posed by Bouveret, Endriss, and Lang (ECAI 2010). We also propose fairness concepts that always suggest a non-empty set of assignments with meaningful fairness properties. Among these concepts, optimal proportionality and optimal weak proportionality appear to be desirable fairness concepts.
	\end{abstract}


\maketitle

	\section{Introduction}

	A basic yet widely applicable problem in computer science and economics is to allocate discrete objects to agents given the preferences of the agents over the objects. The setting is referred to as the \emph{assignment problem} or the \emph{house allocation problem}~\citep[see, \eg][]{ACMM05a,BBL14a,DeHi88a,Gard73b, Manl13a,Wils77a,Youn95b}. 
	In this setting, there is a set of agents $N=\{1,\ldots, n\}$, a set of objects $O=\{o_1,\ldots, o_{m}\}$ with each agent $i\in N$ expressing ordinal preferences $\pref_i$ over $O$. The goal is to allocate the objects among the agents in a fair or optimal manner without allowing transfer of money. The assignment problem is a fundamental setting within the wider domain of \emph{fair division} or \emph{multiagent resource allocation}~\citep{CDE+06a}. The model is applicable to many resource allocation or fair division settings where the objects may be public houses, school seats, course enrollments, kidneys for transplant, car park spaces, chores, joint assets of a divorcing couple, or time slots in schedules. Fair division has become a major area in AI research in the last
decade, and especially the last few years~\citep[see, \eg][]{Aziz14a,BEL10a,BoLa08a,BoLa11a,BoLe14a,BFM+12a,CDE+06a,CLPP11a,DHKP11a,KPG13a,Proc09a}. 

	In this paper, we consider the fair assignment of indivisible objects.
	Two of the most fundamental concepts of fairness are \emph{envy-freeness} and \emph{proportionality}. Envy-freeness requires that no agent considers that another agent's allocation would give him more utility than his own. Proportionality requires that each agent should get an allocation that gives him at least $1/n$ of the utility that he would get if he was allocated all the objects. When agents' ordinal preferences are known but utility functions are not given, then ordinal notions of envy-freeness and proportionality need to be formulated.
	We consider a number of ordinal fairness concepts. Most of these concepts are based on the \emph{stochastic dominance (SD)} relation which is a standard way of comparing fractional/randomized allocations. An agent prefers one allocation over another with respect to the SD relation if he gets at least as much utility from the former allocation as the latter for all cardinal utilities consistent with the ordinal preferences.
	Although this paper is restricted to discrete assignments, using stochastic dominance to define fairness concepts for discrete assignments turns out to be fruitful. The fairness concepts we study include \textit{SD envy-freeness, weak SD envy-freeness, possible envy-freeness, SD proportionality,} and \emph{weak SD proportionality.} 
	We consider the problems of computing a discrete assignment that satisfies some ordinal notion of fairness if one exists, and the problems of verifying whether a given assignment satisfies the fairness notions.

	\paragraph{Contributions}

	We present a systematic way of formulating fairness properties in the context of the assignment problem. The logical relationships between the properties are proved. Interestingly, our framework leads to new solution concepts such as weak SD  proportionality that have not been studied before. The motivation to study a range of fairness
	properties is that, depending on the situation, only some of them are
	achievable. In addition, only some of them can be computed efficiently. In order to find fairest achievable assignment, one can start by checking whether there exists a fair assignment for the strongest notion of fairness. If not, one can try the next fairness concept that is weaker than the one already checked. 

	We present a comprehensive study of the computational complexity of computing fair assignments under ordinal preferences. In particular, we present a polynomial-time algorithm to check whether an SD  proportional exists even when agents may express indifferences.
	The algorithm generalizes the main result of \citep{PrWo12a} (Theorem 1) who focused on strict preferences. 
	For the case of two agents, we obtain a polynomial-time algorithm to check whether an SD  envy-free assignment exist. The result generalizes Proposition 2 in \citep{BEL10a} in which preferences over objects were assumed to be strict.
	For a constant number of agents, we propose a polynomial-time algorithm to check whether a weak SD  proportional assignment exists. As a corollary, for two agents, we obtain a polynomial-time algorithm to check whether a weak SD  envy-free or a possible envy-free assignment exists. Even for an unbounded number of agents, if the preferences are strict, we characterize the conditions under which a weak SD  proportional assignment exists. 
	We show that the problems of checking whether possible envy-free, SD envy-free, or weak SD envy-free assignments exist are NP-complete. The result for possible envy-freeness answers an open problem posed in \citep{BEL10a}. Our computational results are summarized in Table~\ref{table:summary-fair}.

We show that our two main algorithms can be extended to the case where agents have different entitlements over the objects or if we additionally require the assignment to be Pareto optimal.
Our study highlights the impacts of the following settings: $i)$ randomized/fractional versus discrete assignments, $ii)$ strict versus non-strict preferences, and $iii)$ multiple objects per agent versus a single object per agent.
	
	Since the fairness concepts we introduce may not be guaranteed to exist, we suggest possible ways to extend the fairness concepts. Firstly, we consider the problem of maximizing the number of agents for whom the corresponding fairness constraint is satisfied. A criticism of this approach is that there can still be agents who are completely dissatisfied. We then consider an alternative approach in which the proportionality constraints is weakened in a natural and gradual manner. We refer to the concepts as \emph{optimal proportionality} and \emph{optimal weak proportionality}. The fairness concepts are not only attractive but we show that an optimal proportional assignment can be computed in polynomial time and an optimal weak proportional assignment can be computed in polynomial time for a constant number of agents.

\begin{figure}
\newlength{\hsep}
\setlength{\hsep}{1cm}
\newlength{\vsep}
\setlength{\vsep}{1.5cm}
\tikzset{
  >=latex,
  equivalent/.style={double},
  inclusion/.style={<-},
}
	\scalebox{0.86}{
\begin{tikzpicture}
  \draw node[draw] (sdef) {SD EF};
  \draw (sdef.east) ++ (\hsep, 0) node[anchor=west, draw] (nef) {Necessary EF};
  \draw (nef.east) ++ (\hsep, 0) node[anchor=west, draw] (ncef) {NC EF};

  \draw (ncef.south) ++ (0, -\vsep) node[draw] (sdprop) {SD Prop};
  \draw (sdprop.east) ++ (\hsep, 0) node[anchor=west, draw] (nprop) {Necessary Prop};

  \draw (sdprop.south) ++ (0, -\vsep) node[draw] (wsdprop) {Weak SD Prop};
  \draw (nprop.south) ++ (0, -\vsep) node[anchor=west, draw] (pprop) {Possible Prop};

  \draw (sdef.south) ++ (-0.5\hsep, -\vsep) node[draw] (pef) {Possible EF};

  \draw (pef.south) ++ (-0.5\hsep, -\vsep) node[anchor=east,draw] (wsdef) {weak SD EF};
  \draw (wsdef.east) ++ (\hsep, 0) node[anchor=west, draw] (pcef) {PC EF};

  \draw[equivalent] (wsdprop) -- (pprop);
  \draw[equivalent] (wsdef) -- (pcef);
  \draw[equivalent] (sdprop) -- (nprop);
  \draw[equivalent] (sdef) -- (nef);
  \draw[equivalent] (nef) -- (ncef);
  \draw[equivalent] (sdef) to[bend left] (ncef);

  \draw[inclusion] (wsdef) -- (pef);
  \draw[inclusion] (pcef) -- (pef);
  \draw[inclusion] (wsdprop) -- (nprop);
  \draw[inclusion] (wsdprop) -- (sdprop);
  \draw[inclusion] (pprop) -- (sdprop);
  \draw[inclusion] (pprop) -- (nprop);
  \draw[inclusion] (sdprop) -- (sdef);
  \draw[inclusion] (sdprop) -- (nef);
  \draw[inclusion] (sdprop) -- (ncef);
  \draw[inclusion] (nprop) -- (sdef);
  \draw[inclusion] (nprop) -- (nef);
  \draw[inclusion] (nprop) -- (ncef);
  \draw[inclusion] (pef) -- (sdef);
  \draw[inclusion] (pef) -- (nef);
  \draw[inclusion] (pef) -- (ncef);
\end{tikzpicture}
}
\caption{Inclusion relationships between fairness concepts. 
Envy-freeness is abbreviated as EF and Proportionality is abbreviated as Prop. Possible completion is abbreviated as PC. Necessary completion is abbreviated as NC. An arrow represents inclusion. For example, every SD envy-free outcome is also SD proportional. Double lines represent equivalence. For example, SD EF and Necessary EF are equivalent.}\label{fig:inclfigure}
\end{figure}
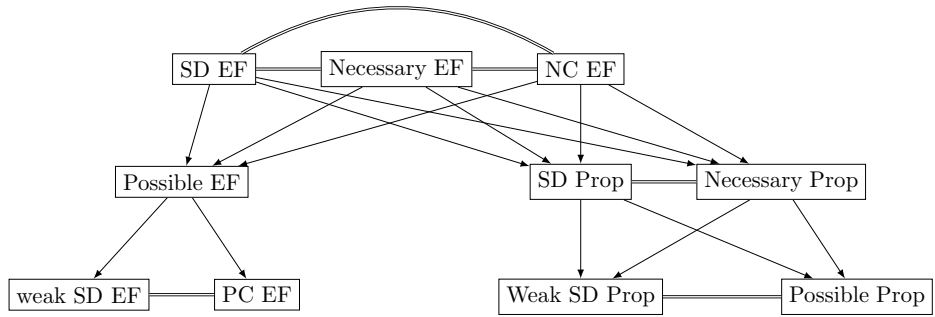

	\begin{table}[h!]
		\centering
			\scalebox{0.9}{
	\centering
	\begin{tabular}{ll}
	\toprule
	\multirow{2}{*}{Weak SD proportional}&\textbf{in P for strict prefs~(Th.~\ref{th:strict-weak-SD  prop})}\\
	&\textbf{in P for constant $n$~(Th.~\ref{th:weak-SD  prop-inP})}\\
	\midrule
	SD proportional&\textbf{in P~(Th.~\ref{th:SD  prop-inP})}\\
	\midrule 
	\multirow{3}{*}{Weak SD envy-free}&\multirow{3}{*}\textbf{NP-complete (Th.~\ref{th:envy-free-is-hard})}\\
&in P for strict prefs~\citep{BEL10a}\\
	&\textbf{in P for $n=2$ (Cor.~\ref{cor:weak-envy-2agents})}\\
	\midrule
	\multirow{3}{*}{Possible envy-free}&\textbf{NP-complete (Th.~\ref{th:envy-free-is-hard})}\\
	&\textbf{in P for strict prefs}\\
	&\textbf{in P for $n=2$ (Cor.~\ref{cor:weak-envy-2agents})}\\
	\midrule
	\multirow{2}{*}{SD envy-free}&NP-complete even for strict prefs~\citep{BEL10a}\\
	&\textbf{in P for $n=2$ (Cor.~\ref{cor:2agents-SD  envy})}\\
	\bottomrule
	\end{tabular}
	}
	\caption{Complexity of checking the existence of a fair assignment of indivisible goods for $n$ agents and $m$ objects.  The results in bold are from this paper.}
	\label{table:summary-fair}
	\end{table}

%

	\section{Related work}

	Proportionality and envy-freeness are two of the most established fairness concepts. Proportionality dates back to at least the work of \citet{Stei48a} in the context of cake-cutting. It is also referred to as \emph{fair share guarantee} in the literature~\citep{Moul03a}.
	A formal study of envy-freeness in microeconomics can be traced back to the work of \citet{Fole67a}.

	The computation of fair discrete assignments has been intensely studied in the last decade within computer science. In many of the papers considered, agents express cardinal utilities for the objects and the goal is to compute fair assignments~\citep[see \eg][]{BeDa05a, BoLa08a, BoLe14a, DeHi88a, Golo05a,LMMS04a, NRR13a, PrWa14a}. 
	A prominent paper is that of \citet{LMMS04a} in which algorithms for approximately envy-free assignments are discussed. It follows from \citep{LMMS04a} that even when two agents express cardinal utilities, checking whether there exists a proportional or envy-free assignment is NP-complete. A closely related problem is the \emph{Santa Claus problem} in which the agents again express cardinal utilities for objects and the goal is to compute an assignment which maximizes the utility of the agent that gets the least utility~\citep[see \eg][]{AsSa10a,BeDa05a,FGM14a,NRR13a}. 
	Just as in \citep{BEL10a,PrWo12a}, we consider the setting in which agents only express ordinal preferences over objects. There are some merits of considering this setting. Firstly, ordinal preferences require elicitation of less information from the agents. Secondly, some of the weaker ordinal fairness concepts we consider may lead to positive existence or computational results. Thirdly, some of the stronger ordinal fairness concepts we consider are more robust than the standard fairness concepts. 
	Fourthly, when the exchange of money is not possible, mechanisms that elicit cardinal preferences may be more susceptible to manipulation because of the larger strategy space.
	Finally, it may be the case that cardinal preferences are simply not available.

There are other papers in fair division in which agents  explicitly express ordinal preferences over sets of objects rather than simply expressing preferences over objects. For these more expressive models, the computational complexity of computing fair assignments is either even higher~\citep{CDE+06a,KBKZ09a} or representing preferences require exponential space~\citep{Aziz14a,BKK12a}. In this paper, we 
restrict agents to simply express ordinal preferences over objects. 
Some papers assume ordinal preferences but superimpose a cardinal utilities via some scoring function~\citep[see \eg][]{BEF03a}. However, this approach does not allow for indifferences in a canonical way and has led to negative complexity results~\citep{BBL14a,DaSc14a,GKK10a}. \citet{GKK10a} assumed that agents have lexicographic preferences and tried to maximize the lexicographic signature of the worst off agents. However the problem is NP-hard if there are more than two equivalence classes.


	The ordinal fairness concepts we consider are SD envy-freeness; weak SD  envy-freeness; possible envy-freeness; SD  proportionality; and weak SD  proportionality. 
	Not all of these concepts are new but they have not been examined systematically for discrete assignments. SD  envy-freeness and weak SD  envy-freeness have been considered in the randomized assignment domain~\citep{BoMo01a} but not the discrete domain. ~\citet{BoMo01a} referred to SD  envy-freeness and weak SD  envy-freeness as envy-freeness and weak envy-freeness.
	SD  envy-freeness and weak SD  envy-freeness have been considered implicitly for discrete assignments but the treatment was axiomatic~\citep{BEF03a,BEF01a}. Mathematically equivalent versions of SD  envy-freeness and weak SD envy-freeness have been considered by \citet{BEL10a} but only for strict preferences. They referred to them as necessary (completion) envy-freeness and possible (completion) envy-freeness.
	A concept equivalent to SD  proportionality was examined by \citet{PrWo12a} but again only for strict preferences. \citet{PrWo12a} referred to weak SD proportionality simply as ordinal fairness.
	 Interestingly, weak SD  or possible proportionality has not been studied in randomized or discrete settings (to the best of our knowledge).

	Envy-freeness  is well-established in fair division, especially cake-cutting. 
	Fair division of goods has been extensively studied within economics but
	in most of the papers, either the goods are divisible or agents are
	allowed to use money to compensate each other~\citep[see \eg][]{Vari74a}.
 In the model we consider, we do not allow money transfers. 

	\section{Preliminaries}

	An assignment problem is a triple $(N,O,\pref)$ such that $N=\{1,\ldots, n\}$ is a set of agents, $O=\{o_1,\ldots, o_m\}$ is a set of objects, and the preference profile $\pref=(\pref_1,\ldots, \pref_n)$ specifies for each agent $i$ his complete and transitive preference $\pref_i$ over $O$.
	Agents may be indifferent among objects.
	We denote $\pref_i: E_i^1,\dots,E_i^{k_i}$ for each agent $i$ with equivalence classes in decreasing order of preferences.
	Thus, each set $E_i^j$ is a maximal equivalence class of objects among which agent $i$ is indifferent, and $k_i$ is the number of equivalence classes of agent $i$. If an equivalence class is a singleton $\{o\}$, we list the object $o$ in the list without the curly brackets. In case each equivalence class is a singleton, the preferences are said to be \emph{strict}.
For any set of objects $O'\subseteq O$, $\max_{\pref_i}(O')=\{o\in O' \midd o \pref_i o' \text{ for each } o'\in O'\}$ and $\min_{\pref_i}(O')=\{o\in O' \midd o' \pref_i o \text{ for each } o'\in O'\}$.

	A fractional assignment $p$ is a $(n\times m)$ matrix $[p(i)(o_j)]$ such that $ p(i)(o_j) \in [0,1]$ for all $i\in N$, and $o_j\in O$, and $\sum_{i\in N}p(i)(o_j)= 1$ for all $j\in \{1,\ldots, m\}$.
	The value $p(i)(o_j)$ represents the probability of object $o_j$ being allocated to  agent $i$. Each row $p(i)=(p(i)(o_1),\ldots, p(i)(o_m))$ represents the allocation of agent $i$. 
	The set of columns correspond to the objects $o_1,\ldots, o_m$.
	A fractional assignment is \emph{discrete} if $p(i)(o)\in \{0,1\}$ for all $i\in N$ and $o\in O$.

\begin{example}\label{example:assign}
	Consider an assignment problem $(N,O,\pref)$ where $|N|=2$, $O=o_1,o_2,o_3,o_4$ and the preferences of the agents are as follows
		\begin{align*}
		1:&\quad o_1,o_2,o_3,o_4\\
		2:&\quad o_2,o_3,o_1,o_4
	\end{align*}
	Then,
	\[p=\begin{pmatrix}
			1&0&0&1\\
		  0&1&1&0
			\end{pmatrix}
		\]
	is a discrete assignment in which agent $1$ gets $o_1$ and $o_4$ and agent $2$ gets $o_2$ and $o_3$.
\end{example}

	A \emph{uniform assignment} is a fractional assignment in which each agent gets $1/n$-th of each object. Although we will deal with discrete assignments, the fractional uniform assignment is useful in defining some fairness concepts. Similarly, we will use the SD relation to define relations between assignments. Our algorithmic focus will be on computing discrete assignments only even though concepts are defined using the framework of fractional assignments.

Informally, an agent `SD prefers' one allocation over another if for each object $o$, the former allocation gives the agent at least as many objects that are  at least as preferred as $o$ as the latter allocation. More formally, given two fractional assignments $p$ and $q$, $p(i) \succsim_i^{SD} q(i)$, i.e., agent $i$ \emph{SD~prefers} allocation $p(i)$ to allocation $q(i)$ if 
	\[\sum_{o_j\in \set{o_k\midd o_k\succsim_i o}}p(i)(o_j) \ge \sum_{o_j\in \set{o_k\midd o_k\succsim_i o}}q{(i)(o_j)} \text{ for all } o\in O.\] He \emph{strictly SD prefers} $p(i)$ to $q(i)$ if $p(i) \succsim_i^{SD} q(i)$ and $\neg [q(i) \succsim_i^{SD} p(i)]$. Although each agent $i$ expresses ordinal preferences over objects, he could have a private cardinal utility $u_i$ consistent with $\pref_i$: $u_i(o)\geq u_i(o') \text{ if and only if } o\pref_i o'.$
	The set of all utility functions consistent with $\pref_i$ is denoted by $\mathcal{U}(\pref_i)$. 	When we consider agents' valuations according to their cardinal utilities, then we will assume additivity, that is $u_i(O')=\sum_{o\in O'}u_i(o)$ for each $i\in N$ and $O'\subseteq O$.

	An assignment $p$ is \emph{envy-free} if the total utility each agent $i$ gets for his allocation is at least the utility he would get if he had any another agent's allocation: 
	\[u_i(p(i))\geq u_i(p(j)) \text{ for all } j\in N.\]
	Note that we sometimes interpret a discrete allocation $p(i)$ as a set, namely the set of objects allocated to agent $i$.
	An assignment is \emph{proportional} if each agent gets at least $1/n$-th of the utility he would get if he got all the objects: \[u_i(p(i))\geq u_i(O)/n.\]

		Note that we require that the assignment is complete, that is, each object is allocated. In the context of fractional assignments, an assignment is complete if no fraction of an object is unallocated.
In the absence of this requirement a null assignment is obviously envy-free. On the other hand a null assignment is not proportional.  

			When allocations are discrete and when agents may get more than one object, we will also consider preference relations over sets of objects. One way of extending preferences over objects to preferences over sets of objects is via the \emph{responsive set extension}~\citep{BBP04a}.
In the \emph{responsive set extension}, preferences over objects are extended to preferences over sets of objects in such a way that a set in which an object is replaced by a more preferred object is more preferred.
Formally, for each agent $i\in N$, his preferences $\pref_i$ over $O$ are extended to his preferences $\pref_i^{RS}$ over $2^O$ via the responsive set extension as follows. For all $S\subset O$, for all $o\in S$, for all $o'\in O\setminus S$,
\begin{align*}
S &\succsim_i^{RS} (S\setminus \{o\})\cup \{o'\} &\text{ if } o'\mathrel{\succsim_i} o,\text{ and}\\
S &\succ_i^{RS} S\setminus \{o\}.
\end{align*}
Equivalent, we say that $p(i) \mathrel{\pref_i^{RS}} q(i)$ if and only if  there is an injection $f$ from $q(i)$ to $p(i)$ such that for each $o\in q(i)$, $f(o)\pref_i o$.


\begin{theorem}\label{prop:sd,rs,util}
	For discrete assignments $p$ and $q$, the following are equivalent.
	\begin{enumerate}
		\item \label{item:sd}	$p(i) \mathrel{\pref_i^{SD}} q(i)$.
				\item \label{item:util}  $\forall u_i\in \mathcal{U}(\pref_i) \text{, } u_i(p(i))\geq u_i(q(i)).$ 
		\item  \label{item:rs} $p(i) \mathrel{\pref_i^{RS}} q(i)$. 
\end{enumerate}
	\end{theorem}
	\begin{proof}
Firstly, \ref{item:sd} and \ref{item:util} are known to be equivalent~\citep[see \eg][]{ABS13a,Cho12a,KaSe06a}.

We now show that \ref{item:rs} implies \ref{item:util}.
If $p(i) \mathrel{\pref_i^{RS}} q(i)$, then we know that for each object allocated to $i$ in $q(i)$, there is an injection which maps the object to an object in $p(i)$ which is at least as preferred by $i$. Hence, for each~$u_i\in \mathcal{U}(\pref_i)$, we have that $u_i(p(i))\geq u_i(q(i)).$ 

We now show that \ref{item:sd} implies \ref{item:rs}.
Assume that $p(i) \mathrel{\not\pref_i^{RS}} q(i)$ .
Consider a bipartite graph $G=(q(i)\mathrel{\cup} p(i),E)$ where $\{o,o'\}\in E$ if $o\in q(i)$, $o'\in p(i)$, and $o'\succsim_i o$. Since $p(i) \mathrel{\not\pref_i^{RS}} q(i)$,  $G$ does not have a matching saturating $q(i)$. 
Then by Hall's theorem, there exists a set $O'\subseteq q(i)$ such that $|N(O')|<|O'|$ where $N(O')$ denote the neighborhood of $O'$. 
Consider an object $o\in \min_{\pref_i}(O')$. We can assume without loss of generality that $O'$ is maximal so that each $o^*\in q(i)$ such that $o^*\pref_i o$ is in $O'$ because this only increases the difference $|O'|-|N(O')|$. Note that $O'$ is then $\{o'\midd o'\pref_i o\} \cap p(i)$ and $N(O')$ is $\{o'\midd o'\pref_i o\} \cap q(i)$.
Since, $|N(O')|<|O'|$, we have that 
			\[ \left|\{o'\midd o'\pref_i o\} \cap p(i)\right| < \left|\{o'\midd o'\pref_i o\} \cap q(i)\right|.\]
%
	But then $p(i) \mathrel{\not\pref_i^{SD}} q(i)$.
\end{proof}

	\section{Fairness concepts under ordinal preferences}

	We now define fairness notions that are independent of the actual cardinal utilities of the agents. The fairness concepts are defined for fractional assignments. Since discrete assignments are special cases of fractional assignments, the concepts apply just as well to discrete assignments. For algorithmic problems, we will only consider those assignments that are discrete. The fairness concepts that are defined are with respect to the SD and RS relations as well as by quantifying over the set of utility functions consistent with the ordinal preferences.

	\begin{table}[t!]
		\centering
			\scalebox{0.9}{
	\centering
	\begin{tabular}{lll}
	\toprule
	SD envy-free&necessary envy-free&necessary completion envy-free\\
	\midrule
	SD proportional&necessary proportional&\\
			\midrule
	&possible envy-free&\\
			\midrule
	weak SD envy-free&&possible completion envy-free\\
			\midrule
		weak SD proportional&possible proportional&\\
	\bottomrule
	\end{tabular}
	}
	\caption{Equivalence between different fairness concepts introduced in the literature. All concepts in the same row are equivalent. For example, weak SD proportional is equivalent to possible proportional.}
	\label{table:summary-equiv}
	\end{table}

		\paragraph{Proportionality}

		\begin{enumerate}
		\item 
		\begin{enumerate}
			\item \emph{Weak SD proportionality}: An assignment $p$ satisfies \emph{weak SD proportionality} if no agent strictly SD prefers the uniform assignment to his allocation: \[\neg[(1/n,\ldots,1/n) \succ_i^{SD} p(i)] \text{ for all } i\in N.\]

			\item \emph{Possible proportionality}: An assignment satisfies \emph{possible proportionality} if for each agent, there are cardinal utilities consistent with his ordinal preferences such that his allocation yields him as at least as much utility as he would get under the uniform assignment:
			\[\text{For each } i\in N, \text{ there exists } u_i\in \mathcal{U}(\pref_i) \text{ such that } u_i(p(i))\geq u_i(O)/n.\]

			\end{enumerate}

			\item 
			\begin{enumerate}
				\item \emph{SD proportionality}: An assignment $p$ satisfies \emph{SD proportionality} if each agent SD prefers his allocation to the allocation under the uniform assignment: \[p(i)\pref_i^{SD} (1/n,\ldots,1/n) \text{ for all } i\in N.\]

				\item \emph{Necessary proportionality}: An assignment satisfies \emph{necessary proportionality} if it is proportional for all cardinal utilities consistent with the agents' preferences.\footnote{\citet{PrWo12a} referred to necessary proportionality as ``ordinal fairness''.}
		\[\text{For each } i\in N, \text{ and for each } u_i\in \mathcal{U}(\pref_i), u_i(p(i))\geq u_i(O)/n.\]
			\end{enumerate}

			\end{enumerate}

			\paragraph{Envy-freeness}
			\begin{enumerate}

				\item
					\begin{enumerate}

					\item \emph{Weak SD envy-freeness}: An assignment $p$ satisfies \emph{weak SD envy-freeness} if no agent strictly SD prefers someone else's allocation to his: \[\neg[p(j)\succ_i^{SD} p(i)] \text{ for all } i,j\in N.\]

					\item \emph{Possible envy-freeness}: An assignment satisfies \emph{possible envy-freeness} if for each agent, there are cardinal utilities consistent with his ordinal preferences such that his allocation yields him as at least as much utility as he would get if he was given any other agent's allocation. 
			\[\text{For all } i\in N, \exists u_i\in \mathcal{U}(\pref_i \nolinebreak) \text{ such that } u_i(p(i))\geq u_i(p(j)) \text{ for all } j\in N\]

						\item \emph{Possible completion envy-freeness}: An assignment satisfies \emph{possible completion envy-freeness}~\citep{BEL10a} 						if for each agent, there exists a preference relation of the
												agent over sets of objects that is a weak order consistent with the
												responsive set extension such that the agent weakly prefers his allocation over
						the allocations of other agents.
												The concept has also been referred to as not ``envy-ensuring''~\citep{BEF01a}.
						
						\end{enumerate}

					\item
								\begin{enumerate}
					\item \emph{SD envy-freeness}: An assignment $p$ satisfies \emph{SD envy-freeness} if each agent SD prefers his allocation to that of any other agent: 
	\[p(i) \pref_i^{SD} p(j) \text{ for all } i,j\in N.\]

							\item \emph{Necessary envy-freeness}: An assignment satisfies \emph{necessary envy-freeness} if it is envy-free for all cardinal utilities consistent with the agents' preferences.
			\[\text{For each } i,j\in N, \text{ and for each } u_i\in \mathcal{U}(\pref_i) \text{, } u_i(p(i))\geq u_i(p(j)).\]

										\item \emph{Necessary completion envy-freeness}: An assignment satisfies \emph{necessary completion envy-freeness}~\citep{BEL10a} if for each agent, and each total order consistent with 
the \emph{responsive set extension} of the agents, each agent weakly prefers his allocation to any other agents' allocation.
										The concept has also been referred to as \emph{not envy-possible}~\citep{BEF01a}.
								\end{enumerate}
								
							We consider the assignment problem in Example~\ref{example:assign} to illustrate some of the fairness notions. 
								
								\begin{example}
									Consider an assignment problem $(N,O,\pref)$ where $|N|=2$, $O=o_1,o_2,o_3,o_4$ and the preferences of the agents are as follows
										\begin{align*}
										1:&\quad o_1,o_2,o_3,o_4\\
										2:&\quad o_2,o_3,o_1,o_4
									\end{align*}
	
									Consider the discrete assignment $p$ in which agent $1$ gets $o_1$ and $o_4$ and agent $2$ gets $o_2$ and $o_3$.
The assignment $p$ is not SD proportional or SD envy-free because the fairness constraints for agent $1$ are not satisfied. However, $p$ is weak SD proportional, possible envy-free, and weak SD envy-free.
			\end{example}

				\end{enumerate}

		Possible completion envy-freeness and necessary completion envy-freeness were simply referred to as possible and necessary envy-freeness in \citep{BEL10a}. We will use the former terms to avoid confusion.

	\section{Relations between fairness concepts}

	In this section, we highlight the inclusion relationships between fairness concepts (see Figure~\ref{fig:inclfigure}).
%
Based on the connection between the SD relation and utilities (Theorem~\ref{prop:sd,rs,util}), we obtain the following equivalences. The equivalences are also summarized in Table~\ref{table:summary-equiv}.

		%
		\begin{theorem}\label{thm:equiv}
					For any number of agents and objects, 
					\begin{enumerate}
							\item  Weak SD proportionality and possible proportionality are equivalent;
					\item \label{item:sdprop-necprop} SD proportionality and necessary proportionality are equivalent;
														\item \label{item:wssdef-compl} weak SD envy-freeness and possible completion envy-freeness are equivalent;
												\item \label{item:sdeff-neceff} SD envy-freeness, necessary envy-freeness and necessary completion envy-freeness are equivalent.
					\end{enumerate}

		\end{theorem}
		\begin{proof}
			We deal with each case separately.
			\begin{enumerate}
	\item The statement follows directly from the characterization of the SD relation.
	\item The statement follows directly from the characterization of the SD relation.

	\item If an assignment is weak SD envy-free, then each agent either SD prefers his allocation over another agent's allocation or finds them incomparable. In case of incomparability, the relation can be completed with the agent's own allocation being more preferred. Thus the assignment is also possible completion envy-free. If an assignment is possible completion envy-free, then either an agent prefers his allocation over another agent's allocation with respect to the responsive set extension or finds them incomparable with respect to the responsive set extension. Hence each agent either SD prefers his allocation over another agent's allocation or finds them incomparable. Thus the assignment is also weak SD envy-free.

	\item It follows from Theorem~\ref{prop:sd,rs,util} that SD envy-freeness and necessary envy-freeness are equivalent. We now prove that SD envy-freeness and necessary completion envy-freeness are equivalent.
Note that an agent SD prefers his allocation over other agents' allocation if and only if he prefers his allocation with respect to the responsive set extension over other agents' allocation. 
\end{enumerate}
		\end{proof}

		It is well-known that when an allocation is complete and utilities are additive, envy-freeness implies proportionality. Assume that an assignment $p$ is envy-free. Then for each $i\in N$, $u_i(p(i))\geq u_i(p(j))$ for all $j\in N$.  Thus, $n\cdot u_i(p(i))\geq \sum_{j\in N}u_i(p(j))=u_{i}(O).$ Hence $u_i(p(i))\geq u_{i}(O)/n$. We can also get similar relations when we consider stronger and weaker notions of envy-freeness and proportionality.

		\begin{theorem}\label{remark:imply}
		The following relations hold between the fairness concepts defined. 
	\begin{enumerate}
		\item SD envy-freeness implies SD proportionality.
		\item SD proportionality implies weak SD proportionality.
		\item \label{item:pefwsdp} Possible envy-freeness implies weak SD proportionality.
		\item Possible envy-freeness implies weak SD envy-freeness.
	
			\end{enumerate}
		\end{theorem}
		\begin{proof}
				We deal with the cases separately.
			\begin{enumerate}
				\item \emph{SD envy-freeness implies SD proportionality}. 
				Assume an assignment $p$ satisfies SD envy-freeness.
				Then, by Theorem~\ref{thm:equiv}\ref{item:sdeff-neceff}, it satisfies envy-freeness for all utilities consistent with the ordinal preferences. If an assignment satisfies envy-freeness for particular cardinal utilities, it satisfies proportionality for the same cardinal utilities. Therefore, $p$ satisfies proportionality for all cardinal utilities consistent with the ordinal preferences. Hence, due to Theorem~\ref{thm:equiv}\ref{item:sdprop-necprop}, it implies that $p$ satisfies SD proportionality.
				\item \emph{SD proportionality implies weak SD proportionality}.
				Assume an assignment $p$ does not satisfy weak SD proportionality. Then, there exists some agent $i\in N$ such that $(1/n,\ldots, 1/n)\succ_i^{SD} p(i)$. But this implies that $\neg[p(i)\pref_i^{SD} (1/n,\ldots, 1/n)]$. Hence $p$ is not SD proportional.
				\item \emph{Possible envy-freeness implies weak SD proportionality}.
				Assume an assignment $p$ is not weak SD proportional. By Theorem \ref{thm:equiv}, $p$ is not possible proportional. Let $i\in N$ be an agent such that for all $u_i\in \mathcal{U}(\pref_i)$ we have that $u_i(p(i)) < u_i(O)/n$. But then, for each $u_i\in \mathcal{U}(\pref_i)$ there exists an agent $j\in N$ such that $u_i(p(i)) < u_i(p(j))$, otherwise $n\cdot u_i(p(i))\ge u_i(O)$. Hence $p$ is not possible envy-free.
				\item \emph{Possible envy-freeness implies weak SD envy-freeness}.
				Assume that an assignment $p$ is not weak SD envy-free. Therefore there exist $i,j\in N$ such that
$p(j)\succ_i^{SD} p(i).$ Due to Theorem~\ref{prop:sd,rs,util}, we get that for each
				$u_i\in \mathcal{U}(\pref_i) \text{, } u_i(p(j))> u_i(p(i)).$ Hence $p$ is not possible envy-free. 	
\end{enumerate}	
		\end{proof}

		We also highlight certain equivalences for the special case of two agents.
		
			\begin{theorem}\label{remark:imply2}
	For two agents, 
	\begin{enumerate}
		\item proportionality is equivalent to envy-freeness; 
		\item SD proportionality is equivalent to SD envy-freeness; 
		\item weak SD proportionality and possible envy-freeness are equivalent; and
		\item weak SD envy-freeness and weak SD proportionality are equivalent.
	\end{enumerate}
\end{theorem}
\begin{proof}
	We deal with the cases separately while assuming $n=2$. Since $n=2$, for any agent $i$, we will denote by $-i$ the other agent.
	\begin{enumerate}
		\item \emph{Proportionality is equivalent to envy-freeness}. Since envy-freeness implies proportionality, we only need to show that for two agents proportionality implies envy-freeness. Assume that an assignment is not envy-free. Then,
			\begin{align*}
			u_i(p(i)) &< u_i(p(-i)) \quad\Leftrightarrow\\
			2 \cdot u_i(p(i)) &< u_i(p(i))+u_i(p(-i)) \quad\Leftrightarrow\\
			u_i(p(i)) &< \frac{u_i(p(i))+u_i(p(-i))}{2}=u_i(O)/2.
		\end{align*}
		\item \emph{SD proportionality is equivalent to SD envy-freeness}.
	We note that for $n=2$, if an assignment satisfies envy-freeness for particular cardinal utilities, it satisfies proportionality for those cardinal utilities. Moreover, if an assignment is SD proportional, it satisfies proportionality for all cardinal utilities, hence it satisfies envy-freeness for all cardinal utilities and hence it satisfies SD envy-freeness.
		\item  \emph{Weak SD proportionality and possible envy-freeness are equivalent}.
		 By Theorem~\ref{remark:imply}\ref{item:pefwsdp}, possible envy-freeness implies weak SD proportionality. If an assignment satisfies weak SD proportionality, then there exist cardinal utilities consistent with the ordinal preferences for which proportionality is satisfied. Hence for $n=2$, there exist cardinal utilities consistent with the ordinal preferences for which envy-freeness is satisfied, which means that the assignment satisfies possible envy-freeness.
		%
	\item \emph{Weak SD envy-freeness and weak SD proportionality are equivalent}. We have already shown that weak SD proportionality implies possible envy-freeness for $n=2$, and that possible envy-freeness implies weak SD envy-freeness. Therefore, it is sufficient to prove that weak SD envy-freeness implies weak SD proportionality.
	Assume that an assignment $p$ is not weak SD proportional.
	Then, there exists at least one agent $i\in \{1,2\}$ such that 
		\[\frac{\card{\bigcup_{j=1}^kE_i^j}}{2} \geq \card{(\bigcup_{j=1}^kE_i^j)\cap p(i)}\] 
for all $k\in \{1,\ldots, k_i\}$ and
		\[\frac{\card{\bigcup_{j=1}^kE_i^j}}{2} > \card{(\bigcup_{j=1}^kE_i^j)\cap p(i)}\] 
for some $k\in \{1,\ldots, k_i\}$.
But this implies that 
		\[{\card{\bigcup_{j=1}^kE_i^j \cap p(-i)}} \geq \card{(\bigcup_{j=1}^kE_i^j)\cap p(i)}\] 
for all $k\in \{1,\ldots, k_i\}$ and
		\[{\card{\bigcup_{j=1}^kE_i^j \cap p(-i)}} > \card{(\bigcup_{j=1}^kE_i^j)\cap p(i)}\] 
for some $k\in \{1,\ldots, k_i\}$.
Thus $p(-i) \succ_i^{SD} p(i)$ and hence $p$ is not weak SD envy-free.
	\end{enumerate}
\end{proof}

	In the next examples, we show that some of the inclusion relations do not hold in the opposite direction and that some of the solution concepts are incomparable. Firstly, we show that SD proportionality does not imply weak SD envy-freeness.

	\begin{example}
		SD proportionality does not imply weak SD envy-freeness.
		Consider the following preference profile:
		\begin{align*}
		&1:\quad \{a, b, c\}, \{d, e, f\}\\
		&2:\quad\{a,b,c,d,e,f\}\\
		&3:\quad\{a,b,c,d,e,f\}
	\end{align*}
		The allocation that gives $\{a,d\}$ to agent $1$, $\{b,c\}$ to agent $2$ and $\{e,f\}$ to agent $3$ is SD  proportional. However it is not weak SD  envy-free since agent $1$ is envious of agent $2$. 
		Hence it also follows that SD proportionality does not imply possible envy-freeness or SD envy-freeness.


	\end{example}

Next, we show that 	weak SD envy-freeness neither implies possible envy-freeness nor weak SD proportionality.

	\begin{example}
	\label{example:weaknotpossible}
	Weak SD envy-freeness neither implies possible envy-freeness nor weak SD proportionality.
	Consider an assignment problem in which $N=\{1,2,3\}$, and there are $4$ copies of $A$, $6$ copies of $B$, $1$ copy of $C$ and $1$ copy of $D$. Let the preference profile be as follows. 
		\begin{align*}
		&1:\quad A,B,C,D\\
		&2:\quad\{A\}, \{B,C,D\}\\
		&3:\quad\{B\}, \{A,C,D\}.
	\end{align*}

	\begin{table}[h!]
	\centering	
	\begin{tabular}{l|llll}
	\small
	&$A$&$B$&$C$&$D$\\\midrule
	$1$&$1$&$1$&$1$&$1$\\
	$2$&$3$&$0$&$0$&$0$\\
	$3$&$0$&$5$&$0$&$0$\\
	\end{tabular}
		\caption{Discrete assignment $p$ in Example~\ref{example:weaknotpossible}}
	\label{table:assignment}
	\end{table}

	Clearly $p$, the assignment specified in Table~\ref{table:assignment} is weak SD envy-free. Assume that $p$ is also possible envy-free. Let $u_1$ be the utility function of agent $1$ for which he does not envy agent $2$ or $3$. Let $u_1(A)=a$; $u_1(B)=b$; $u_1(C)=c$; and $u_1(D)=d$. Since $A\succ_1 B \succ_1 C \succ_1 D$, we get that 
	\begin{equation}\label{eq:abcd}
	a>b>c>d.
	\end{equation}
	Since $p$ is possible envy-free, $u_1(p(1))\geq u_1(p(2)) \text{ iff } a+b+c+d\geq 3a 
	\text{ iff } a\leq \frac{b+c+d}{2} \text{which implies } a<\frac{3b}{2}$.
	Since $p$ is possible envy-free, 
	$u_1(p(1))\geq u_1(p(3))$ iff  $a+b+c+d\geq 5b$ iff $a+c+d\geq 4b$. Since $a>b>c>d$, it follows that $a> 2b.$
	%
	This is a contradiction since both $a<\frac{3b}{2}$ and $a\geq 2b$ cannot hold.

	Now we show that weak SD envy-freeness does not even imply weak SD proportionality. Assignment $p$ is weak SD envy-free. If it were weak SD proportional then there exists a utility function $u_1$ such that $u_1(a)+u_1(b)+u_1(c)+u_1(d)\geq \frac{4u_1(a)+6u_1(b)+u_1(c)+u_1(d)}{3}$ which means that $\frac{u_1(a)}{3}+u_1(b)\leq \frac{2u_1(c)}{3}+\frac{2u_1(d)}{3}$ which is equivalent to $\frac{a}{3}+b\leq \frac{2c}{3}+\frac{2d}{3}$. But this is not possible because of \eqref{eq:abcd}.
	\end{example}

Since, we have shown that weak SD envy-freeness is not equivalent to possible envy-freeness, and since we showed in Theorem~\ref{thm:equiv}\ref{item:wssdef-compl} that weak SD envy-freeness is equivalent to possible completion envy-freeness, 
this means that possible envy-freeness and possible completion envy-freeness are also not equivalent to each other.
We now point out that possible envy-freeness does not imply SD proportionality.

	\begin{example}
		Possible envy-freeness does not imply SD proportionality.
	Consider an assignment problem with two agents with preferences 
	$1: \{a\},\{b,c\}$ and $2: \{a,b,c\}$.
	Then the assignment in which $1$ gets $a$ and $2$ gets $b$ and $c$ is possible envy-free. However it is not SD  proportional, because agent $1$'s allocation does not SD  dominate the uniform allocation. 
	\end{example}

			Finally, we note that all notions of proportionality and envy-freeness are trivially satisfied if randomized assignments are allowed by giving each agent $1/n$ of each object. As we show here, achieving any notion of proportionality is a challenge when outcomes need to be discrete.


	Next, we study the existence and computation of fair assignments.
			Even the weakest fairness concepts like weak SD proportionality may not be possible to achieve: consider two agents with identical and strict preferences over two objects. This problem remains even if $m$ is a multiple of $n$.

			\begin{example}
			A discrete weak SD proportional assignment may not exist even if $m$ is a multiple of $n$.
			Consider the following preferences: 
				\begin{align*}
				1: \quad \{a_1,a_2,a_3,a_4\}, \{b_1,b_2\}\\
				2: \quad \{a_1,a_2,a_3,a_4\}, \{b_1,b_2\}\\
				3: \quad \{a_1,a_2,a_3,a_4\}, \{b_1,b_2\}
			\end{align*}
	If all agents get 2 objects, then
	those agents that have to get at least one object from $\{b_1,b_2\}$ will get an allocation that is strictly SD  dominated by $(1/3,\ldots, 1/3)$.
	Otherwise, at least one agent gets at most one object, and is therefore strictly SD dominated by the uniform assignment.

	If $m$ is not a multiple of $n$, then an even simpler example shows that a weak SD proportional assignment may not exist. Consider the case when all agents are indifferent among all objects. Then the agent who gets less objects than $m/n$ will get an allocation that is strictly SD  dominated by $(1/n,\ldots, 1/n)$. 

	\end{example}

			
			\section{Computational Complexity}

In this paper, we consider the natural computational question of checking whether a discrete fair assignment exists and if it does exist then to compute it. The problem of verifying whether a (discrete or fractional) assignment is fair is easy for all the notions we defined.

		\begin{remark}\label{remark:verify}
	It can be verified in time polynomial in $n$ and $m$ whether an assignment is fair for all notions of fairness considered in the paper.
	For possible envy-freeness, a linear program can be used to find the `witness' cardinal utilities of the agents. 
	\end{remark}

	\begin{remark}\label{th:const-objects}
	For a constant number of objects, it can be checked in polynomial time whether a fair discrete assignment exists for all notions of fairness considered in the paper. 
	This is because the total number of discrete assignments is $n^m$.
	\end{remark}
	
	We note that if the assignment is not required to be discrete, then even SD envy-freeness can be easily achieved~\citep{KaSe06a}. Finally, we have the following necessary condition for SD proportional and hence for SD envy-free assignments.
	
	\begin{theorem}\label{th:m=nc}
		If $p$ is a discrete SD proportional assignment, then $m$ is a multiple of $n$ and each agent gets $m/n$ objects. 
		\end{theorem}
		\begin{proof}
			If $p$ is an SD proportional assignment, then the following constraint is satisfied for each agent $i\in N$.
			\begin{equation*}
			\left|p(i)\cap O\right|\geq \frac{|O|}{n}= \frac{m}{n}.
			\end{equation*}
		Each agent must get $m/n$ objects. If $p$ is discrete, each agents gets $m/n$ objects only if $m$ is a multiple of $n$.
			\end{proof}

	\subsection{SD proportionality}

	In this subsection, we show that it can be checked in polynomial time whether a discrete SD  proportional assignment exists even in the case  of indifferences. The algorithm is via a reduction to the problem of checking whether a bipartite graph admits a feasible $b$-matching.

	Let $H=(V_H,E_H)$ be an undirected graph with vertex capacities $b:V_H\rightarrow \natz$ and edge capacities $c:E_H\rightarrow \natz$ where $\natz$ is the set of natural numbers including zero. 
	Then, a \emph{$b$-matching} of~$H$ is a function $m:E_H\rightarrow\natz$ such that $\sum_{e\in \{e'\in E_H \midd v\in e'\}}m(e)\leq b(v)$ for each $v\in V_H$, and $m(e)\leq c(e)$ for all $e\in E_H$. 
	The \emph{size} of the $b$-matching $m$ is defined as $\sum_{e\in E_H}m(e)$. We point out that if $b(v)=1$ for all $v\in V_H$, and $c(e)=1$ for all $e\in E_H$ then a maximum size $b$-matching is equivalent to a maximum cardinality matching. 
	In a $b$-matching problem with upper \emph{and} lower bounds, there further is a 
	function $a:V_H\rightarrow\natz$. A feasible $b$-matching then is a function $m:E_H \rightarrow \natz$ such that $a(v)\leq\sum_{e\in\{e'\in E_H \midd v\in e'\}}m(e)\leq b(v)$ for each $v\in V_H$, and $m(e)\leq c(e)$ for all $e\in E_H$. 
	If $H$ is bipartite, then the problem of computing a maximum weight feasible $b$-matching with lower and upper bounds can be solved in strongly polynomial time~\citep[Chapter 35,][]{Schr03a}.


	\begin{theorem}\label{th:SD  prop-inP}
	It can be checked in polynomial time whether a discrete SD  proportional assignment exists even if agents are allowed to express indifference between objects.
	\end{theorem}
	\begin{proof}
	Consider $(N,O,\pref)$. 
	If $m$ is not a multiple of $n$, then by Theorem~\ref{th:m=nc}, no discrete SD  proportional assignment exists.
	In this case, in each discrete assignment $p$, there exists some agent $i\in N$ who gets less than $m/n$ objects. Thus, the following does not hold: $p(i) \succsim_i^{SD} (1/n,\ldots, 1/n)$.  Hence we can now assume that $m$ is a multiple of $n$ i.e., $m=nc$ where $c$ is a constant.
	We reduce the problem to checking whether a feasible $b$-matching exists for a graph $G=(V,E)$. Recall that $k_i$ is the number of equivalence classes of agent $i$.
	For each agent $i$, and for each $\ell \in \{1,\ldots, k_i\}$ 
	we introduce a vertex $v_i^{\ell}$. 
	For each $o\in O$, we create a corresponding vertex with the same name. Now,
	$V=\{ v_i^1,\ldots, v_i^{k_i} \midd i\in N\} \cup O.$
	The graph $G$ is bipartite with independent sets $O$ and $V\setminus O$. 
	Let us now specify the edges of $G$:

			\begin{itemize}
				\item for each $i\in N$, $\ell\in \{1,\ldots, k_i\}$ and $o\in O$ we have that $\{v_i^{\ell},o\}\in E$ if and only if $o\in \bigcup_{j=1}^{\ell} E_i^j$.
			\end{itemize}

			We specify the lower and upper bounds of each vertex:

			\begin{itemize}
				\item $a(v_i^{\ell})=\ceil{\frac{\sum_{j=1}^\ell \card{E_i^{j}}}{n}} - \sum_{j=1}^{\ell-1} a(v_i^j)$
				and $b(v_i^\ell) = \infty$ for each $i\in N$ and $\ell\in \{1,\ldots, k_i\}$;
				\item $a(o)=b(o)=1$ for each $o\in O$.
			\end{itemize}

				For each edge $e\in E$, $c(e)=1$.

	Now that $(V,E)$ has been specified, we check whether a feasible $b$-matching exists.
	If so, we allocate an object $o$ to an agent $i$ if the edge incident to $o$ that is included in the matching is incident to a vertex corresponding to an equivalence class of agent $i$.
	We claim that a discrete SD proportional assignment exists if and only if a feasible $b$-matching exists.
	If a feasible $b$-matching exists, then each $o\in O$ is matched so we have a complete assignment. For each agent $i\in N$, and for each $E_i^\ell$, an agent is allocated at least $\ceil{{\sum_{j=1}^\ell \card{E_i^{j}}}/{n}}$ objects of the same or more preferred equivalence class. Thus, the assignment is SD proportional. 

	On the other hand if a discrete SD proportional assignment $p$ exists, then $p(i)\pref_i^{SD} (1/n,\ldots, 1/n)$ implies that for each equivalence class $E_i^\ell$, an agent is allocated at least $\ceil{{\sum_{j=1}^\ell \card{E_i^{j}}}/{n}}$ objects from the same or more preferred equivalence class as $E_i^\ell$. Hence there is a $b$-matching in which the lower bound of each vertex of the type $v_i^\ell$ is met. For any remaining vertices $o\in O$ that have not been allocated, they may be allocated to any agent. Hence a feasible $b$-matching exists.
	\end{proof}

	\subsection{Weak SD proportionality}

	In the previous subsection, we examined the complexity of checking the existence of SD proportional discrete assignments. In this section we consider weak SD proportionality. 

	\begin{theorem}
		\label{th:strict-weak-SD  prop}
	For strict preferences, a weak SD proportional discrete assignment exists if and only if one of two cases holds:
	\begin{enumerate}
\item	$m=n$ and  it is possible to allocate to each agent an object that is not his least preferred object;
\item $m>n$.
	\end{enumerate}

		Moreover, it can be checked in polynomial time whether a weak SD proportional discrete assignment exists when agents have strict preferences.
	\end{theorem}
	\begin{proof}

	If $m<n$, at least one agent will not get any object. Hence there exists no  weak SD  proportional discrete assignment. Hence $m\geq n$ is a necessary condition for the existence of a weak SD proportional discrete assignment.


	Let us consider the case of $m=n$. Clearly each agent needs to get one object. If an agent $i$ gets an object that is not the least preferred object $o'$, then his allocation $p(i)$ is weak SD  proportional. The reason is that $\sum_{o\succ o'}p(i)(o)=1>\card{\{o\midd o\succ o'\}}/n.$
	Hence the following does not hold: $(1/n,\ldots, 1/n) \succsim_i^{SD} p(i)$. On the other hand, if $i$ gets the least preferred object, his allocation is not weak SD  proportional since $(1/m,\ldots, 1/m) \succ_i^{SD} p(i)$.	 
	Hence, we just need to check whether there exists a discrete assignment in which each agent gets an object that is not least preferred.
	This can be solved as follows. We construct a graph $(V,E)$ such that $V=N\cup O$ and for all $i\in N$ and $o\in O$, $\{i,o\}\in E$ if and only if $o\notin\min_{\pref_i}(O)$. We just need to check whether $(V,E)$ has a perfect matching. If it does, the matching is a weak SD  proportional discrete assignment.

	If  $m\geq n+1$, we show that a weak SD proportional discrete assignment exists. 
	Allocate the most preferred object to the agents in the following order
	$1,2,3,\ldots n, n ,n-1\ldots, 1, \dots$.
	Then each agent $i\in \{1,\ldots, n-1\}$ gets in the worst case his $i$-th most preferred object. This worst case occurs if agents preceding $i$ pick the $i-1$ most preferred objects of agent $i$. Even in this worst case, since $1>i/n$, we have that the allocation of agents in $\{1,\ldots, n-1\}$ is weak SD  proportional. As for agent $n$, in the worst case he get his $n$-th and $n+1$st most preferred objects. Since $2\geq \frac{n}{n} + \frac{1}{n}$,  by Lemma~\ref{lemma:weaksdprop} we get that the allocation of agent $n$ is also weak SD  proportional. 
	This completes the proof.
	\end{proof}


	Indifferences result in all sorts of challenges. Some arguments that we used for the case for strict preferences do not work for the case of indifferences. 
	The case of strict preferences may lead one to wrongly assume that given a sufficient number of objects, a weak SD proportional discrete assignment is guaranteed to exist. However, if agents are allowed to express indifference, this is not the case. Consider the case where $m=nc+1$ and each agent is indifferent among each of the objects. Then there exists no weak SD  proportional discrete assignment because some agent will get fewer than $m/n$ objects.
	We first present a helpful lemma which follows directly from the definition of weak SD  proportionality.

	\begin{lemma}
	\label{lemma:weaksdprop}
	An assignment $p$ is weak SD proportional if and only if for each $i\in N$,
	\begin{enumerate}
		\item $\sum_{o'\succsim o} p(i)(o')> \card{\{o'\midd o'\succsim o\}}/n$ for some $o\in O$; or 
			\item  $\sum_{o'\succsim o} p(i)(o')\geq \card{\{o'\midd o'\succsim o\}}/n$ for all $o\in O$.
	\end{enumerate}
	\end{lemma}

	We will use Lemma~\ref{lemma:weaksdprop} in designing an algorithm to check whether a weak SD proportional discrete assignment exists when agents are allowed to express indifference. 

	\begin{theorem}\label{th:weak-SD  prop-inP}
	For a constant number of agents, it can be checked in polynomial time whether a weak SD  proportional discrete assignment exists even if agents are allowed to express indifference between objects.
	\end{theorem}
	\begin{proof}

		Consider $(N,O,\pref)$. We want to check whether a weak SD proportional discrete assignment exists. 
		By Lemma~\ref{lemma:weaksdprop}, this is equivalent to checking whether there exists a discrete assignment $p$, where for each $i\in N$, one of the following $k_i$ conditions holds: for $l\in \{1,\ldots, k_i\}$, 
					\begin{align}
		&\sum_{o\in \bigcup_{j=1}^{l} E_i^j} p(i)(o)> \frac{\card{\bigcup_{j=1}^{l} E_i^j}}{n}
							\end{align}

	\noindent		
			or the following  $(k_i+1)$-st condition holds
				\begin{align}
		p(i)\sim_i^{SD} (1/n,\ldots, 1/n).
						\end{align}

					The $(k_i+1)$-st condition only holds if each $\card{E_i^j}$ is a multiple of $n$ for $j\in \{1,\ldots, k_i\}$.


						We need to check whether there exists a discrete assignment in which for each agent one of the $k_i+1$ conditions is satisfied. In total there are $\prod_{i=1}^n (k_i+1)$ different ways in which the agents could be satisfied. We will now present an algorithm to check if there exists a feasible weakly SD  proportional discrete assignment in which for each agent $i$, a certain condition among the $k_i+1$ conditions is satisfied. Since $n$ is a constant,  the total number of combinations of conditions is polynomial.

	We define a bipartite graph $G=(V,E)$ whose vertex set is initially empty.
	For each agent $i$, if the condition number is $\ell \in \{1,\ldots, k_i\}$ 
	then we add a vertex $v_i^{\ell}$. If the condition number is $k_i+1$, then we add $k_i$ vertices --- $B_i^j$ for each $E_i^j$ where $j\in \{1,\ldots, k_i\}$. 
	For each $o\in O$, we add a corresponding vertex with the same name.
	The sets $O$ and $V\setminus O$ will be independent sets in $G$. We now specify the edges of $G$.

			\begin{itemize}
				\item $\{v_i^{\ell},o\}\in E$ if and only if $o\in \bigcup_{j=1}^{\ell} E_i^j$ for each $i\in N$, $\ell\in \{1,\ldots, k_i\}$ and $o\in O$.
				\item $\{B_i^j,o\}\in E$ if and only if $o\in E_i^j$ for each $i\in N$, $j\in \{1,\ldots, k_i\}$, and $o\in O$.
			\end{itemize}

			We specify the lower and upper bounds of each vertex.

			\begin{itemize}
				\item $a(v_i^{\ell})=\floor{\frac{\card{\bigcup_{j=1}^{\ell} E_i^j}}{n}}+1$ and $b(v_i^{\ell})=\infty$ for each $i\in N$ and $\ell\in \{1,\ldots, k_i\}$;
				\item $a(B_i^j)=b(B_i^j)=\frac{\card{E_i^j}}{n}.$ for each $B_i^j$;
				\item $a(o)=b(o)=1$ for each $o\in O$.
			\end{itemize}

	For each edge $e\in E$, $c(e)=1$.
	For each $n$-tuple of satisfaction conditions, we construct the graph as specified above and then check whether there exists a feasible $b$-matching. A weak SD proportional discrete assignment exists if and only if a feasible $b$-matching exists for the graph corresponding to at least one of the $\prod_{i=1}^n (k_i+1)$ combinations of conditions. Since $\prod_{i=1}^n (k_i+1)$ is polynomial if $n$ is a constant and since a feasible $b$-matching can be checked in strongly polynomial time, we can check the existence of a weak SD  proportional discrete assignment in polynomial time.
	\end{proof}

	\subsection{Envy-freeness}

In this section, we examine the complexity of checking whether an envy-free assignment exists or not. Our positive algorithmic results for SD proportionality and weak SD proportionality help us obtain algorithms for SD envy-freeness and weak SD envy-freeness when $n=2$. 

From Theorem~\ref{th:SD  prop-inP}, we get the following corollary.

	\begin{corollary}\label{cor:2agents-SD  envy}
		For two agents, it can be checked in polynomial time whether a discrete SD envy-free assignment exists even if agents are allowed to express indifference between objects.
	\end{corollary}
	\begin{proof}
		For two agents, SD  proportionality implies SD  envy-freeness, and by Theorem \ref{remark:imply},
		SD envy-freeness implies SD proportionality.
	\end{proof}

	Corollary~\ref{cor:2agents-SD  envy} generalizes Proposition~10 of \citep{BEL10a} which stated that for two agents and \emph{strict} preferences, it can be checked in polynomial time whether a necessary envy-free discrete assignment exists.
	
Similarly, from Theorem~\ref{th:weak-SD  prop-inP}, we get the following corollary.
	
	\begin{corollary}\label{cor:weak-envy-2agents}
			For two agents, it can be checked in polynomial time whether a weak SD  envy-free or a possible envy-free discrete assignment exists.
	\end{corollary}
	\begin{proof}
		For two agents, weak SD  proportional is equivalent to weak SD  envy-free and possible envy-free (Theorem~\ref{remark:imply2}).
	\end{proof}

We prove that checking whether a (weak) SD  envy-free or possible envy-free discrete assignment exists is NP-complete. The complexity of the second problem was mentioned as an open problem in \citep{BEL10a}.
	\citet{BEL10a} showed that the problem of checking whether a necessary envy-free discrete assignment exists is NP-complete. The statement carries over to the more general domain that allows for ties. We point out that if agents have identical preferences, it can be checked in linear time whether an SD  envy-free discrete assignment exists even when preferences are not strict. Identical preferences have received special attention within fair division~\citep[see \eg][]{BrFi00a}.

	\begin{theorem}\label{th:envy-identical}
		For agents with identical preferences, an SD  envy-free discrete assignment exists if and only if each equivalence class is a multiple of $n$.
	\end{theorem}

	Even $n$ is not constant but preferences are strict, it can be checked in time linear in $n$ and $m$ whether a complete weak SD envy-free discrete assignment exists. This follows from an equivalent result in \citep{BEL10a} for possible completion envy-freeness and the fact that weak SD envy-freeness is equivalent to possible completion envy-freeness (Theorem~\ref{thm:equiv}\ref{item:wssdef-compl}). We use similar arguments as \citet{BEL10a} for possible envy-freeness.

	%
	%
	
	\begin{theorem}\label{th:strict-possible-SD  envy}
	For strict preferences, it can be checked in time linear in $n$ and $m$ whether a possible envy-free discrete assignment exists. 
	\end{theorem}
	\begin{proof}
		We reuse the arguments in the proof of \citep[][Proposition 4]{BEL10a}.
	Let the number of distinct top-ranked objects be $k$. 
	If $m<2n-k$, then there is at least one agent who receives one object that is not his top-ranked $o$ and no further items. Thus he necessarily envies the agent who received $o$ and hence there cannot exist a possible envy-free discrete assignment.
	If $m\geq 2n-k$, then we run the following algorithm.
	(1) For each of the $k$ top-ranked objects, allocate it to an agent that ranks it first. Denote by $N'$ the set of agents that have not yet received an object, and order them arbitrarily. (2) Go through the $n - k$ agents in $N'$ in ascending order and ask them to claim their most preferred item from those still available. (3) Go through the agents in $N'$ again, but this time in descending order, and ask them to claim their most preferred item from those still available. 
	The agents who got their most preferred object do not envy any other agent if they have sufficiently high utility for their most preferred object.
For the remaining agents, who have received two objects each, no agent $i$ strictly SD prefers another agent $j$'s current allocation: even if $j$ (who had an earlier first turn) received a more preferred first object, $i$ strictly prefers his second object to $j$'s second object (in case $j$ received a second object). 
Therefore, there exist cardinal utilities consistent with the ordinal preferences where the agents in $N'$ put high enough utility for the second object they get so that they are not envious of other agents even if the other agent gets all the unallocated objects. Therefore the unallocated objects can be allocated in an arbitrary manner among the remaining agents and the resulting complete discrete assignment is still possible envy-free.
	\end{proof}

	\citet{BEL10a} mentioned the complexity of possible completion envy-freeness for the case of indifferences as an open problem. We present a reduction to prove that for \emph{all} notions of envy-freeness considered in this paper, checking the existence of a fair discrete assignment is NP-complete. 

	\begin{theorem}\label{th:envy-free-is-hard}
		The following problems are NP-complete:
		\begin{enumerate}
			\item check whether there exists a weak SD  envy-free (equivalently possible completion envy-free) discrete assignment,
			\item check whether there exists a possible envy-free discrete assignment, and
			\item check whether there exists an SD envy-free discrete assignment.
		\end{enumerate}
	\end{theorem}
	\begin{proof}
	Membership in NP is shown by Remark~\ref{remark:verify}. To show hardness we use a reduction from X3C (Exact Cover by 3-sets).
	In X3C, the input is a ground set $S$ and a collection $C$ containing 3-sets of elements from $S$, and the question is whether there exists
	a subcollection $X\subseteq C$ such that each element of $S$ is contained in exactly one of the 3-sets in $X$. X3C is known to be NP-complete~\citep{Karp72a}.
	Consider an instance $(S,C)$ of X3C where $S=\{s_1,...,s_{3q}\}$ and $C=\{c_1,...,c_l\}$. Without loss of generality, $l\geq q$. 
	We construct the following assignment problem $(N,O,\pref)$ where
	$N=\{a_1,...,a_{40l}\}$ is partitioned into three sets $N_1, N_2$ and $N_3$ with $|N_1|=l$, $|N_2|=30l$, $|N_3|=9l$ and
	$O=\{o_1,...,o_{120l}\}$ is partitioned into three sets $O_1, O_2$ and $O_3$ with $|O_1|=3l$, $|O_2|=90l$ and $|O_3|=27l$.
	The set $O_1$ is partitioned into two sets, $O^S_1$ and $O^B_1$, the first one corresponding to the set of elements of $S$ in the X3C instance and the second being a `buffer' set. We have $|O^S_1|=3q$ and $|O^B_1|=3l-3q$.
	We associate each $c_j \in C$ with the $j$-{th} agent in $N_1$. With each $c_j \in C$ we also associate nine consecutive agents in $N_3$. 
	The preferences of the agents are defined as follows:\\

	\noindent
	$i:O_2 \cup c_i,(O_1\backslash c_i) \cup O_3$ for $i \in N_1$\\
	$i:O_2,O_1 \cup O_3$ for all $i \in N_2$\\
	$i:f(i),O\backslash f(i)$ for $i \in N_3$\\

	The function $f:N_3 \rightarrow 2^O$ is such that it ensures the following properties: 
	For each of the three elements $e$ of $c_j$, three out of the nine agents associated with $c_j$ list $e$ as a second choice object, and list $c_j\backslash \{e\}$ as first choice objects. Let us label these three agents $a_1$, $a_2$ and $a_3$.
	The sets of objects $f(a_1)$, $f(a_2)$ and $f(a_3)$ each exclude a distinct $\frac{1}{3}$ of the buffer objects $O^B_1$.
	$f(a_1)$,  $f(a_2)$ and $f(a_3)$ each contain $2/3$ of the elements of $O^B$. The $1/3$  elements that are excluded from each of these sets must be distinct i.e., elements that $f(a_1)$ does not contain are contained in $f(a_2)$ and $f(a_3)$, and vice versa.

	For each $i \in N_3$, $f(i) \cap (O_2 \cup O_3)=O_f$.
	Let $O_f$ contain $\frac{2}{3}$ of the elements of $O_2$ and $\frac{2}{3}$ of the elements of $O_3$. \
	Consider a discrete assignment that is weak SD envy-free or possible envy-free or SD envy-free. We can make the following observations:

	\begin{enumerate}
	\item Agents in $N_2$ are allocated all objects from $O_2$ and none from $O\backslash O_2$.
	To show this, first consider the case where $30l$ or more objects from $O_2$ are assigned to $N \backslash N_2$. In this case, at least one agent in $N_2$ is envious of an agent from $N\backslash N_2$: there will be an agent $b_1$ in $N \backslash N_2$ with three or more objects from $O_2$, and there will be an agent $b_2$ in $N_2$ with at most three elements, at most two of which are from $O_2$. This is because if an agent has more than three objects, another has at most two
	and if they all have three, some of those will be objects from $O_1$, and at least one agent from $N_2$ will have a second choice object. For all considered notions of envy-freeness $b_2$ will be envious of $b_1$.

	If $0<z_1<30l$ objects from $O_2$ are assigned to $N \backslash N_2$, we have three cases:
		\begin{enumerate}
	\item $z_2 < z_1$ objects from $O\backslash O_2$ are assigned to $N_2$. In this case an agent from $N_2$ has two or less objects, which implies he will be envious of others in $N_2$.
	\item $z_2=z_1$ objects from $O\backslash O_2$ are assigned to $N_2$. To not be envious of each other agents from $N_2$ will each receive two first choice objects and one second choice object. At least one agent from $N_1$ will receive at least three objects from $O_2$, making agents in $N_2$ envious of him.
	\item $z_2 > z_1$ objects from $O\backslash O_2$ are assigned to $N_2$. In this case all agents from $N_2$ are given three or four objects. If an agent has two, he will be envious as before. There are not enough objects left for each agent in $N\backslash N_2$ to receive three or more objects. Therefore one of these agents, labelled $b_1$ only has two items. Even if those two items are most preferred items, he will be envious of at least one agent in $N_1$ because to any agent in $N\backslash N_2$ the ratio of most preferred items assigned to $N_1$ is higher than $\frac{1}{3}$.
	This implies at least one agent in $N_1$ will have two most preferred items according to $b_1$, and since all in $N_2$ have at least three objects, $b_1$ is envious of that agent.

	\end{enumerate}

	\item Each agent in $N_2$ is allocated exactly three objects. Since as shown above all and only $O_2$ objects go to $N_2$, not all agents in $N_2$ can have four objects. Therefore if one has four, those without four objects will envy him since they value all objects from $O_2$ the same.
	\item Each agent in $N_2 \cup N_1$ has three objects. This is because if an agent in $N_2 \cup N_1$ has four or more objects, another has two or less. The argument in the first observation still applies, and therefore this agent will be envious of at least one agent from $N_2$.
	\item Agents in $N_1$ will not be assigned any objects from $O_3$ since they all consider them to be second choices. To not envy agents in $N_2$ agents in $N_1$ have three of their preferred choices. 
	\item Each agent in $N_2 \cup N_3$ are given two of $N_3$'s common preferred choices, and one of their second choices. This is the only way to avoid envy from an element of $N_3$ to at least one element of $N_2 \cup N_3$: if an element of $N_2$ has two or three of $N_3$ second choices, then another has three preferred choices, and therefore at least one of $N_3$ will be envious of him. If an agent in $N_3$ has three preferred choices, then at least one has only one preferred choice, and will be envious of the agent with three preferred choices.
	\item An agent from $N_1$ does not have objects from $O^S_1$ and also $O^B_1$, since otherwise at least one agent from $N_3$ will be envious of him. This is because of the conditions satisfied by $f$. There are at least three agents in $N_3$ who see the one or two selected elements from the three-set associated to the $N_1$ agent as first choice objects. For any set of elements of size two or less in $O^B_1$, at least one of these three agents considers said set to be composed of first choice object. Therefore, there is at least one agent in $N_3$ who will be envious of an agent in $N_1$ who selects both from $O^S_1$ and $O^B_1$, since he sees this agent as having three preferred choices whilst he only has two (according to the previous observation).
	\end{enumerate}

	If there exists an exact cover of $S$ by a subset of $C$, then there is an SD envy-free discrete assignment since agents corresponding to elements of $C$ used for the cover will be given their preferred items from $O^S_1$ and the others will be given items from $O^B_1$. 
	
	If there does not exist an exact cover of $S$ by a subset of $C$ then there does not exist a weak SD envy-free discrete assignment (equivalently an assignment in which no agent strictly prefers another agent's allocation with respect to responsive preferences). This is because even if all the previous conditions are respected, at least one agent from $N_1$ gets a second choice object and is envious of agents from $N_2$. This follows from the fact that no matter which agents of $N_1$ we assign buffer objects to, the remaining agents are not able to cover $O^S_1$ with their sets of most preferred objects.
	This completes the proof.
	\end{proof}

In view of Theorem~\ref{prop:sd,rs,util}, the proof above also shows that when agents have responsive preferences over sets of objects, then checking whether there exists an envy-free (weak or strong) allocation is NP-complete. Another corollary is that is if agents have cardinal utilities $1$ or $0$ for objects, then checking whether there exists an envy-free assignment is NP-complete. 

\section{Extensions}

In this section, we consider two extensions of our results: \emph{(i)} additionally requiring Pareto optimality and \emph{(ii)} handling varying entitlements. We show that our algorithmic results can be extended in both cases.

\subsection{Additionally requiring Pareto optimality}

We focussed on fairness and only required a weak form of efficiency that each object is allocated. In this subsection, we seek discrete assignments that are both fair and Pareto optimal.

Let $(N,O,\pref)$ be an assignment problem.
A discrete assignment $q$ \emph{Pareto dominates} a discrete assignment $p$
if $q(i) \succsim_i^{SD} p(i)$ for all $i\in N$ and $q(i) \succ_i^{SD} p(i)$ for some $i\in N$.
We also say that $q$ is a \emph{Pareto improvement} over $p$.
A discrete assignment $p$ is \emph{Pareto optimal} if there exists no discrete assignment $q$ that Pareto dominates it.
%

\begin{example}
	Consider the following assignment problem:
	\begin{align*}
		1:&\quad \{a,b,c\}, \{d,e,f\}\\
		2:&\quad \{d,e,f\}, \{a,b,c\}
	\end{align*}
	The discrete assignment that gives $\{b,c,f\}$ to agent $1$ and $\{d,e,a\}$ to agent $2$ is SD proportional. However it is not Pareto optimal since it is SD-dominated by the assignment in which agent $1$ gets $\{a,b,c\}$ and agent $2$ gets $\{d,e,f\}$.
\end{example}


Let $(N,O,\pref)$ be an assignment problem and $p$ be a discrete assignment.
We will create an auxiliary assignment problem and assignment where each agent is allocated exactly one object.
The \emph{clones} of an agent $i\in N$ are the agents in $N_i'=\{i_o \midd o \in O \text{ and } p(i)(o)=1\}$.
The \emph{cloned assignment problem} corresponding to assignment problem $(N,O,\pref)$ and assignment $p$ is $(N',O,\pref')$ such that 
$N'= \bigcup_{i\in N} N_i'$.
and for each $i_o\in N'$, $\pref'_{i_o}=\pref_i$. The \emph{cloned assignment} of $p$ is the discrete assignment $p'$ in which $p'(i_o)(o)=1$ if $p(i)(o)=1$ and $p'(i_o)(o)=0$ otherwise.

A cloned assignment can easily be transformed back into the original assignment where each agent $i\in N$ is allocated all the objects assigned by $p'$ to the clones of $i$.

%
%


\begin{lemma}\label{lemma:cloned}
	A discrete assignment is Pareto optimal if and only if its cloned assignment is Pareto optimal for the cloned assignment problem. 
\end{lemma}
\begin{proof}
	Let $(N,O,\pref)$ be an assignment problem and $p$ be a discrete assignment. We will prove that $p$ is not Pareto optimal if and only if its cloned assignment $p'$ is not Pareto optimal for the cloned assignment problem $(N',O,\pref')$.
	
	For the backwards direction, assume that $p'$ is not Pareto optimal for $(N',O, \allowbreak \pref')$. Then, there exists another discrete assignment $p'^*$ in which each of the cloned agents get at least as preferred an object and at least one agent gets a strictly more preferred object. But if $p'^*$ is transformed into the discrete assignment $p^*$ for the original assignment problem, then $p^*$ Pareto dominates $p$.
	
	For the forward direction, assume that $p$ is not Pareto optimal. Then, there exists another discrete assignment $p^*$ that Pareto dominates it. But this implies that the cloned assignment of $p^*$ also Pareto dominates $p'$ in $(N',O,\pref')$ (modulo name changes among clones).
	\end{proof}
	
	%


\begin{lemma}\label{lemma:po-improve}
	If a discrete assignment is not Pareto optimal, a Pareto improvement that is Pareto optimal can be computed in polynomial time.
\end{lemma}
\begin{proof}
	We first take the assignment problem and the given discrete assignment and construct the corresponding cloned assignment problem and cloned assignment. For such a cloned discrete assignment, a Pareto optimal Pareto improvement can be computed in polynomial time~\citep[see \eg][]{AzKe12a}. 
The updated cloned assignment is then transformed back into an assignment for the original assignment problem.		
		\end{proof}

		\begin{remark}\label{remark:ef-po}
			A Pareto improvement over a weak SD  proportional or SD proportional discrete assignment is weak SD proportional or SD proportional, respectively. Therefore, if a (weak) SD proportional discrete assignment exists then there also exists a Pareto optimal and (weak) SD proportional discrete assignment.
		\end{remark}
	
Based on Lemma~\ref{lemma:po-improve} and Remark~\ref{remark:ef-po} we obtain the following theorems.

\begin{theorem}\label{th:SD  prop-inP-PO}
If a Pareto optimal and SD proportional discrete assignment exists, it can be computed in polynomial time.	
\end{theorem}

\begin{theorem}\label{th:weak-SD  prop-inP-PO}
For a constant number of agents, if a Pareto optimal and weak SD proportional discrete assignment exists, it can be computed in polynomial time.
\end{theorem}

%
%

%

%
%

	\subsection{Unequal entitlements}

Throughout this paper, we assumed that each agent has the same entitlement to the objects. However, it could be the case that an agent $i\in N$ has entitlement $e_i$. There can be various reasons for unequal entitlements. An agent may be given more entitlement for the resources to reward his contributions and effort in obtaining the objects for the set of objects. Entitlements can also be used to model justified demand. For example, if an agent represents a different number of sub-agents, the agent who represents more sub-agents may have more entitlement. Unequal entitlements have been considered in the fair division literature (see, \eg \citep[][page 44]{BrTa96a}.

In the case of unequal entitlements, proportionality and envy-freeness can be redefined:
\[u_i(p(i))\geq \frac{e_i}{\sum_{j\in N}e_j}\sum_{o\in O}{u_i(o)} \text{ for each } i\in N \] for proportionality and 
\[u_i(p(i))\geq \frac{e_i}{e_j}{u_i(p(j))} \text{ for each } i,j\in N\] for envy-freeness. 
For envy-freeness, the idea is that if agent $i$ has half the entitlement of agent $j$, then $i$ will only be envious of agent $j$ if agent's $j$ allocation gives agent $i$ more than twice the utility agent $i$ has for his own allocation.
Just like possible and necessary fairness is defined for equal entitlements, the definitions can be extended for the case of unequal entitlements. Hence possible and necessary proportionality and envy-freeness are natural ordinal notions that can also take into account entitlements.
Our two algorithms for possible and necessary proportionality can also be modified to cater for entitlements by replacing $1/n$ with $\frac{e_i}{\sum_{j\in N}e_j}$ whenever a matching lower bound is specified for a vertex.

	\begin{theorem}\label{th:weak-SD  prop-inP-entitle}
	For a constant number of agents, it can be checked in polynomial time whether a possible proportional discrete assignment exists even if agents have different entitlements.
	\end{theorem}

	\begin{theorem}\label{th:SD  prop-inP-entitle}
	It can be checked in polynomial time whether a necessary proportional discrete assignment exists.
	\end{theorem}

	%

\section{Fairness Concepts that Guarantee the Existence of Fair Outcomes}
\label{sec:guaranteeexist}

We observed that there are instances where even the weakest fairness notions such as weak SD proportionality cannot be guaranteed. Hence the fairness notions considered are not proper solution concepts. In this section, we propose fairness concepts that always suggest a non-empty set of assignments with meaningful fairness properties. 

	\subsection{Maximal and Maximum Fairness}
	\label{sec:maxfair}

We first seek a way out by considering corresponding solution concepts that maximize the number of agents being satisfied with their allocation. The idea has been used in matching theory where for example if a stable matching does not exist, then one may aim to minimize the number of unstable pairs (see, \eg \citep[][]{BMM12a}). 
For each fairness notion $X\in \{\text{SD envy-freeness}, \text{weak SD envy-freeness}, \allowbreak\text{possible }\allowbreak \text{envy-freeness}, \text{SD proportionality}, \text{weak SD proportionality}\}$, we define the following concepts:
	\begin{enumerate}
	\item \emph{Maximum $X$}: a discrete assignment $p$ satisfies Maximum $X$ if it maximizes the total number of agents for which the fairness condition according to $X$ is satisfied.
	\item \emph{Maximal $X$}: a discrete assignment $p$ satisfies Maximal $X$ if the fairness condition according to $X$ cannot be be satisfied for any more agents while maintaining the fairness condition for agents who are satisfied by $p$.
	\end{enumerate}

	The following lemmas are useful in relating the complexity of fairness concept $X$ with \emph{Maximum $X$} and \emph{Maximal $X$}. The proofs are straightforward.

	\begin{lemma}
	If there exists a polynomial-time algorithm to compute a discrete assignment that is \emph{Maximum $X$} then there exists a polynomial-time algorithm to compute a discrete assignment satisfying $X$ if one exists.
	\end{lemma}
	\begin{proof}
	Simply compute the \emph{Maximum $X$} and check whether the fairness condition is satisfied for each agent. If not, then the fairness condition cannot be satisfied for each agent and hence no $X$ discrete assignment exists.
	\end{proof}
	\begin{corollary}
	Computing a maximum SD envy-free discrete assignment is NP-hard. 
	\end{corollary}

	\begin{lemma}
	If there exists a polynomial-time algorithm to check whether a discrete assignment satisfying $X$ exists, then the problem of computing a \emph{Maximal $X$} discrete assignment can also be solved in polynomial time.
	\end{lemma}
	\begin{proof}
		We describe the reduction. 
	Initialize $S$, the set of agents for which the fairness condition can be met, to the empty set.
Check whether there exists an agent $j\in N\setminus S$ that can be moved to $S$ such that there still exists a discrete assignment that satisfies the fairness condition according to $X$ for agents in $S$. If yes, then move $j$ to $S$. Repeat the process until the set $S$ cannot be grown. Hence $S$ is the maximal set of agents that can be satisfied.
	\end{proof}
	\begin{corollary}
	A maximal SD proportional discrete assignment can be computed in polynomial time. 
	\end{corollary}

	\begin{corollary}
	A maximal weak SD proportional discrete assignment can be computed in polynomial time if $n$ is a constant.
	\end{corollary}
	
	Figure~\ref{fig:tnfigure} illustrates the polynomial-time reductions between computational problems for fairness concept $X\in \{$envy-freeness, proportionality, SD envy--freeness, weak SD envy-freeness, possible envy-freeness, weak SD proportionality, SD proportionality$\}$.

		\begin{figure}
					    \centering
					    \scalebox{0.6}{
							\large
						\begin{tikzpicture}
							\tikzstyle{pfeil}=[->,>=angle 60, shorten >=1pt,draw]
							\tikzstyle{onlytext}=[]

						\node[onlytext] (MaximalX) at (0,2) {\Large{\sc MaximalX}};
							\node[onlytext] (ExistsX) at (0,0) {\Large{\sc ExistsX}};
							\node[onlytext] (MaximumX) at (0,-2) {\Large{\sc MaximumX}};

							\draw[pfeil] (MaximalX) to (ExistsX);
							\draw[pfeil] (ExistsX) to (MaximumX);

						\end{tikzpicture}
						}
						\caption{Polynomial-time reductions between computational problems for fairness concept $X\in \{$envy-freeness, proportionality, SD envy--freeness, weak SD envy-freeness, possible envy-freeness, weak SD proportionality, SD proportionality$\}$.}\label{fig:tnfigure}
					  \end{figure}
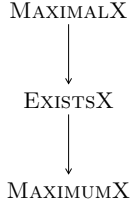

	\subsection{Optimal proportionality}

	A possible criticism of maximal and maximum fairness is that the fairness constraint of each agent is strong enough so that it is not possible to satisfy it for each agent. Hence those agents that do not have their fairness constraints satisfied may not view the assignment as fair from their perspective. To counter this criticism, we weaken the fairness constraint in a uniform way which leads to attractive fairness concepts called \emph{optimal proportionality} and \emph{optimal weak proportionality}. The concepts are similar to egalitarian equivalence rule for continuous resource settings~\citep{PaSc78a}. For continuous settings, an allocation satisfies egalitarian equivalence if each agent is indifferent between his allocation and the reference resource bundle. Since we consider only ordinal preferences, we exploit the SD relations to define suitable concepts. Moreover, since we consider discrete assignments, we relax the requirement of each agent's allocation being equivalent to the reference allocation. 
	
	We say that	an assignment satisfies \emph{$1/\alpha$ proportionality} if 
	\[p(i)\pref_i^{SD} (1/\alpha,\ldots,1/\alpha) \text{ for all } i\in N.\] We note that $1/n$ proportionality is equivalent to SD proportionality.
An assignment satisfies \emph{optimal proportionality} if
	\[p(i)\pref_i^{SD} (1/\alpha,\ldots,1/\alpha) \text{ for all } i\in N.\] for the smallest possible $\alpha$. We will refer to the smallest such $\alpha$ as $\alpha^*$ and call $1/\alpha^*$ as the \emph{optimal proportionality value}.
	
	We point out that Theorem~\ref{th:SD  prop-inP} can be generalized from $1/n$ proportionality to $1/\alpha$ proportionality for any value of $\alpha$:

	\begin{theorem}\label{th:alphaprop-inP}
	It can be checked in polynomial time whether a discrete $1/\alpha$  proportional assignment exists even if agents are allowed to express indifference between objects.
	\end{theorem}
	\begin{proof}
		The algorithm and proof is identical to that of the algorithm in the proof of Theorem~\ref{th:SD  prop-inP}. The only difference is that for each $i\in N$ and $\ell\in \{1,\ldots, k_i\}$, the lower bound of each vertex is set to $a(v_i^{\ell})=\ceil{\frac{\sum_{j=1}^\ell \card{E_i^{j}}}{\alpha}} - \sum_{j=1}^{\ell-1} a(v_i^j)$ instead of $\ceil{\frac{\sum_{j=1}^\ell \card{E_i^{j}}}{n}} - \sum_{j=1}^{\ell-1} a(v_i^j)$.
		\end{proof}

The algorithm in the proof of Theorem~\ref{th:alphaprop-inP} can be used to check the existence of a $1/\alpha$ proportional assignment for different values of $\alpha$. However, among other cases, if $m<n$, then we know that a $1/\alpha$ proportional assignment does not exist for any finite value of $\alpha$. We first characterize the settings that admit a $1/\alpha$ proportional assignment for some finite $\alpha$.

	\begin{lemma}
		$\alpha^*$ is finite iff there exists an assignment in which each agent gets one of his most preferred objects.
		\end{lemma}
		\begin{proof}
			In case each agent gets one of his most preferred objects, proportionality is satisfied for $\alpha^*=mn$.
			If $\alpha^*$ is finite, then each agent's proportionality constraint with respect to the first equivalence class is satisfied. Hence for some  finite $\alpha^*$, 
			\[\left|p(i)\cap E_i^1\right|\geq \frac{|E_i^1|}{\alpha^*}.\]
			Hence each $i$ gets at least one of his most preferred objects. 
			\end{proof}

Since $\alpha$ is a positive real in the interval $(0,\infty]$, it appears that even binary search cannot be used to find the optimal proportional assignment in polynomial time. Next, we show that interestingly we only need to check a polynomial number of values of $\alpha$ to find the optimal proportional assignment.´
	
	\begin{theorem}
		An optimal proportional assignment can be computed in polynomial time even if agents are allowed to express indifference between objects.
		\end{theorem}
	\begin{proof}
		If there exists no assignment in which each agent gets one of his most preferred objects, then $\alpha^*$ is infinite. This means the first proportionality constraint of agents cannot be simultaneously satisfied for a finite $\alpha^*$.
In that case any arbitrary assignment satisfies optimality proportionality with $\alpha^*=\infty$.	
In case there exists an assignment in which each agent gets one of his most preferred objects, then $\alpha^*$ is finite. We show how to compute such an $\alpha^*$ as well as the assignment corresponding to it. 
An assignment $p$ satisfies $1/\alpha$ proportionality if for each $i\in N$ and each $k\in \{1,\ldots, k_i\}$,
\begin{equation}\label{eq:tight}
\left|p(i)\cap \bigcup_{j=1}^k E_i^j\right|\geq \frac{|\bigcup_{j=1}^k E_i^j|}{\alpha}.
\end{equation}
Since $k_i\leq m$, there are in total $mn$ such constraints.
Since the left hand side of each such constraint is an integer, the overall constraint is tight if 
\[\left|p(i)\cap \bigcup_{j=1}^k E_i^j\right|= \ceil{\frac{|\bigcup_{j=1}^k E_i^j|}{\alpha}}.\]

If the value of $\alpha$ is such that no proportionality constraint is tight then this means that $\alpha$ is not optimal. Therefore, we can restrict our attention to those values of $\alpha$ for which at least one of the constraints of the type in \eqref{eq:tight} is tight. When a constraint is tight, both sides of the constraint take one of the values from the set $\{1,\ldots, m\}.$ No constraint can take value $0$ because we know that $\alpha^*$ is finite and that there exists an assignment in which each agent gets one of his most preferred objects. 
For the tight constraint 

\[\left|p(i)\cap \bigcup_{j=1}^k E_i^j\right|= \frac{|\bigcup_{j=1}^k E_i^j|}{\alpha}\in (\ell,\ell+1]\]
for some $\ell\in \{0,\ldots, m-1\}$. This constraint is tight for one of the following values of $\alpha$:

\[\left\{\frac{|\bigcup_{j=1}^k E_i^j|}{\ell+1}\midd \ell\in \{0,\ldots, m-1\}\right\}\]

It follows that if we restrict $\alpha$ to those values for which at least one proportionality constraint is tight, then we just need to consider at most $nm^2$ values of $\alpha$ all of which are rationals. 
For each of these values, we check whether a $1/\alpha$ proportional assignment exists or not. The smallest $\alpha$ for which a $1/\alpha$ proportional assignment exists is the optimal $\alpha=\alpha^*$. The assignment is the optimal proportional assignment. 
\end{proof}

We note that if an SD proportional assignment exists, then it is also an optimal proportional assignment.

\begin{theorem}\label{th:sdprop=optimal}
SD proportionality implies optimal proportionality irrespective of whether the assignments are discrete or not.
	\end{theorem}
	\begin{proof}
		Assume that assignment $p$ is SD proportional not an optimal proportional assignment. Then, we know that it is $1/n$ proportional. Assume for contradiction that $p$ is $1/\alpha$ proportional for $\alpha<n$. But this means that for each $i\in N$, the following constraint is satisfied.
		
		\begin{equation*}
		\left|p(i)\cap O\right|\geq \frac{|O|}{\alpha}.
		\end{equation*}
		
		Since $\alpha<n$, it follows that $\frac{|O|}{\alpha}>\frac{|O|}{n}$. But if each agent gets more than $\frac{|O|}{n}$ objects, then the number of objects is more than $|O|$ which is a contradiction.\end{proof}
		
We show that even if an SD proportional assignment does not exist, an optimal proportional assignment suggests a desirable allocation of objects.
		
				\begin{example}
					Assume that the preferences of the agents are as follows.
						\begin{align*}
						1:&\quad \{o_1,o_2,o_3\}\\
						2:&\quad \{o_1,o_2,o_3\}
					\end{align*}
					Since $m$ is not a multiple of $n$, an SD proportional assignment does not exist. Now consider the assignment $p$ that gives two objects to one agent and one object to the other. It is an optimal proportional assignment where the optimal proportionality value is $1/3$.
					\end{example}

Optimal proportionality seems to be a useful fairness concept for ordinal settings. It guarantees the existence of an assignment that satisfies a fairness notion along the lines of proportionality. If each agent cannot get a most preferred object, then the optimal proportionality value reached is  $1/\infty=0$. If the optimal proportionality value is $0$, then one can modify the preference profile by gradually merging the first few equivalence classes of the agents. If $m>n$, then after merging enough equivalence classes of the agents, one can ensure that the optimal proportional value for the modified profile is finite. An optimal proportional assignment for the modified preference profile stills seems to constitute a desirable and fair assignment for the original preference profile. We note that in contrast to the ordinal setting, when agents have cardinal utilities, even checking whether there exists a proportional assignment is NP-complete~\citep{BoLe14a,LMMS04a}.

		\subsection{Optimal weak proportionality}

		Just like the concept of SD proportionality can be used to define optimal proportionality, weak SD proportionality can be used to define \emph{optimal weak proportionality}.
			We say that	an assignment satisfies \emph{$1/\beta$ weak proportionality} if 
			\[(1/\beta,\ldots,1/\beta) \nsucc_i^{SD}  p(i) \text{ for all } i\in N.\] We note that $1/n$ weak proportionality is equivalent to weak SD proportionality.
An assignment satisfies \emph{optimal weak proportionality} if
\[(1/\beta,\ldots,1/\beta) \nsucc_i^{SD}  p(i) \text{ for all } i\in N.\] for the infimum of the set  $\{\beta\ |\ \text{ $\exists$ a $1/\beta$ weak proportional assignment}\}$.  
We will refer to the infimum as $\beta^*$ and call $1/\beta^*$ as the \emph{optimal weak proportionality value}.
 
	We point out that Theorem~\ref{th:weak-SD  prop-inP} can be generalized from $1/n$ proportionality to $1/\beta$ proportionality for any value of $\beta$:
			
			\begin{theorem}\label{th:alphaweak-SDprop-inP}
			For a constant number of agents, it can be checked in polynomial time whether a $1/\beta$ weak proportional discrete assignment exists even if agents are allowed to express indifference between objects.
			\end{theorem}
			\begin{proof}
				The algorithm and proof is identical to that of the algorithm in the proof of Theorem~\ref{th:weak-SD  prop-inP}. The only difference in the algorithm is in the lower bounds of vertices.
					\begin{itemize}
					\item $a(v_i^{\ell})=\floor{\frac{\card{\bigcup_{j=1}^{\ell} E_i^j}}{\beta}}+1$ and $b(v_i^{\ell})=\infty$ for each $i\in N$ and $\ell\in \{1,\ldots, k_i\}$;
					\item $a(B_i^j)=b(B_i^j)=\frac{\card{E_i^j}}{\beta}.$ for each $B_i^j$;
					\item $a(o)=b(o)=1$ for each $o\in O$.
				\end{itemize}
				\end{proof}

			The algorithm in the proof of Theorem~\ref{th:alphaweak-SDprop-inP} can be used to check the existence of a $1/\beta$ weak proportional assignment for different values of $\beta$. However, among other cases, if $m<n$, then we know that a $1/\beta$ weak proportional assignment does not exist for any finite value of $\beta$. 
			
			\begin{lemma}
				For any assignment setting, $\beta^*\geq 1$.
				\end{lemma}
				\begin{proof}
					Assume for contradiction that  $\beta^*< 1$. This means that there exists a discrete assignment $p$ such that for each agent $i\in N$, either $p(i)\sim_i^{SD} (1/\beta^*,\ldots, 1/\beta^*)$  or
$\left|p(i)\cap \bigcup_{j=1}^k E_i^j\right|> \frac{|\bigcup_{j=1}^k E_i|}{\beta^*}$ for some $k$. 
Note that the former condition is not feasible because it means that $p(i)\succ_i^{RS} O$. For the latter condition, since $\beta^*<1$, therefore $\left|p(i)\cap \bigcup_{j=1}^k E_i^j\right|>|\bigcup_{j=1}^k E_i^j|$. But this is a contradiction.
\end{proof}

Next, we characterize those assignment settings for which $\beta^*$ is finite.

	\begin{lemma}\label{lemma:weakpropvaluecharac}
		$\beta^*$ is finite if and only if  $m\geq n$.
		\end{lemma}
		\begin{proof}
			If $\beta^*$ is finite, then at least one of each agent's weak $\beta^*$ proportionality constraint is satisfied. This implies that each agent gets at least one object which means that $m\geq n$.

Assume that $m\geq n$. Hence there exists an assignment in which each agent gets one object. In the worst case, some agent $i\in N$ gets only one object that is also his least preferred. Even then $p(i)\succ_i^{SD} (1/\beta,\ldots, 1/\beta)$ if $1/\beta <1/m$. Hence $p$ is weak $1/\beta$ weak proportional for any finite $\beta>m$.			
	
			\end{proof}

Next, we show that if the number of agents is constant, an optimal weak proportional assignment can be computed in polynomial time. 
\begin{theorem}
	If the number of agents is constant, a discrete optimal weak proportional assignment can be computed in polynomial time even if agents are allowed to express indifference between objects.
\end{theorem}
	\begin{proof}

If $m<n$, then by Lemma~\ref{lemma:weakpropvaluecharac}, $\beta^*$ is infinite and at least one agent cannot get an object. In that case, any assignment satisfies $1/\infty$ weak proportionality.
If $m\geq n$, then by Lemma~\ref{lemma:weakpropvaluecharac}, $\beta^*$ is finite, 
we show how to compute such a $\beta^*$ as well as the assignment corresponding to it. 
Assume that $\beta^*$ is such that there exists at least one $i\in N$ such that $p(i)\sim_i^{SD} (1/\beta^*,\ldots, 1/\beta^*)$.
In this case, for each $j\in \{1,\ldots, k_i\}$ $i$ gets $\frac{|E_i^j|}{\beta^*}$ objects from $p(i)\cap E_i^j$. Since $|E_i^j|\in \{1,\ldots, m\}$ and $|p(i)\cap E_i^j|\in \{1,\ldots, m\}$,
$\beta^*$ takes one of at most $O(m^2)$ values from set $\{a/b\midd a,b \in \{1,\ldots, m\}, a>b\}$.

Now let us assume that $\beta^*$ is such that there exists no $i\in N$ such that $p(i)\sim_i^{SD} (1/\beta^*,\ldots, 1/\beta^*)$. In that case for each agent $i\in N$, 
an assignment $p$ satisfies $1/\beta$ proportionality if for each $i\in N$ and some $k\in \{1,\ldots, k_i\}$,
\begin{equation}
\left|p(i)\cap \bigcup_{j=1}^k E_i^j\right|> \frac{|\bigcup_{j=1}^k E_i^j|}{\beta}.
\end{equation}

For an arbitrarily small rational $\epsilon>0$, this constraint is equivalent to 
\[\left|p(i)\cap \bigcup_{j=1}^k E_i^j\right|\geq \ceil{\frac{|\bigcup_{j=1}^k E_i^j|}{\beta-\epsilon}}.\]


Since $k_i\leq m$, there are in total $mn$ such constraints.
Since the left hand side of each such constraint is an integer, the overall constraint is tight if 
\begin{equation}\label{eq:tight2}
\left|p(i)\cap \bigcup_{j=1}^k E_i^j\right|= \ceil{\frac{|\bigcup_{j=1}^k E_i^j|}{\beta-\epsilon}}.
\end{equation}

If the value of $\beta-\epsilon$ is such that no proportionality constraint is tight then this means that $\beta-\epsilon$ is not optimal. Therefore, we can restrict our attention to those values of $\beta-\epsilon$ for which at least one of the constraints of the type in \eqref{eq:tight2} is tight. When a constraint is tight, both sides of the constraint take one of the values from the set $\{1,\ldots, m\}.$ 
For the tight constraint, 

\[\left|p(i)\cap \bigcup_{j=1}^k E_i^j\right|= \frac{|\bigcup_{j=1}^k E_i^j|}{\beta-\epsilon}\in (\ell,\ell+1]\]
for some $\ell\in \{0,\ldots, m-1\}$. This constraint is tight for one of the following values of $\beta-\epsilon$:

\[\left\{\frac{|\bigcup_{j=1}^k E_i^j|}{\ell+1}\midd \ell\in \{0,\ldots, m-1\}\right\}\]

It follows that if we restrict $\beta-\epsilon$ to those values for which at least one proportionality constraint is tight, then we just need to consider at most $nm^2$ values of $\beta-\epsilon$ all of which are rationals. We can informally consider these values as the values of $\beta$ because $\beta-\epsilon$ is a tiny perturbation of $\beta$.
For each of the values, we check whether a $1/\beta$ proportional assignment exists or not. 
Similarly, for each of the values from the set $\beta\in \{a/b\midd a,b \in \{1,\ldots, m\}, a>b\}$, we check whether a $\beta$ weak proportional assignment exists or not. The smallest $\beta$ for which a $1/\beta$ assignment exists is $\beta^*$. 
\end{proof}

It remains open whether the theorem above can be generalized to the case where the number of agents is not constant.
We note that whereas an SD proportional assignment is an optimal proportional assignment (Theorem~\ref{th:sdprop=optimal}), a weak SD proportional assignment may not be an optimal weak proportional assignment.


\begin{example}
	Assume that the preferences of the agents are as follows.
		\begin{align*}
		1:&\quad o_1,o_2,o_3,o_4,o_5\\
		2:&\quad \{o_2,o_3\},\{o_1,o_4,o_5\}
	\end{align*}
Note that the discrete assignment $p$ that gives $\{o_2,o_3\}$ to agent $1$ and the other objects to agent $2$ is weak SD proportional. In fact it is not only $1/2$ weak proportional but $(3/5-\epsilon)$
weak proportional where $\epsilon>0$ is arbitrarily small. It is not $1/\beta$ weak proportional for $1/\beta <3/5$.
	
We now consider a discrete assignment $q$, that gives $\{o_1\}$ to agent $1$ and the other objects to agent $2$.
It is easily seen that $q$ is $(1-\epsilon)$ weak proportional where $\epsilon>0$ is arbitrarily small.
Since the weak SD proportional assignment $p$ is not $1/\beta$ weak proportional for $1/\beta <3/5$ but there exists another discrete assignment $q$ that is $(1-\epsilon)$ weak proportional, it shows that 
a weak SD proportional discrete assignment may not be an optimal weak proportional assignment.
	\end{example}

Next, we provide further justification for optimal weak proportionality.
Optimal weak proportionality is equivalent to an established fairness concept called \emph{maximin} defined for a restricted domain. When $n=m$ and preferences of agents are strict, \citet{BrKi05a} defined an assignment satisfying maximin if it maximizes the minimum rank of items that any player receives. We show that for $n=m$ and strict preferences, maximin is equivalent to  optimal weak proportional.

\begin{theorem}
	For $n=m$ and strict preferences, maximin is equivalent to optimal weak proportional.
	\end{theorem}
\begin{proof}
	Assume that each agent gets an object that is at least his $k$-th ranked object. Then the assignment satisfies $1/(k-\epsilon)$ weak proportionality for some arbitrarily small $\epsilon>0$.
	We now assume that an assignment satisfies $1/(k-\epsilon)$ weak proportionality for some arbitrarily small $\epsilon>0$. In that case we know that each agent gets an object that is $k$-th or higher ranked. 
	\end{proof}

	\section{Conclusions}

	We have presented a taxonomy of fairness concepts under ordinal preferences, and identified
	the relationships between the concepts.
	Compared to refinements of the responsive set extension  to define fairness concepts~\citep{BEL10a}, using cardinal utilities and the SD  relation to define fairness concepts not only gives more flexibility (for example reasoning about entitlements) but can also be convenient for algorithm design.

	A problem with the fairness concepts presented in the paper was that the set of fair outcomes is not guaranteed to be non-empty. In view of this
	we highlighted in Section~\ref{sec:guaranteeexist} how alternative notions of maximal and maximum fairness are useful and how the respective computational problems are related to each other. Another possible way to circumvent non-existence of SD proportional or weak SD proportional assignments is to define less stringent notions by replacing the reference vector $(\frac{1}{n}, \ldots, \frac{1}{n})$ with a reference vector with $1/n$ replaced by some constant. We show that an optimal proportional discrete assignment can be computed in polynomial time whereas an optimal weak proportional discrete assignment can be computed in polynomial time if the number of agents is bounded. Both fairness notions appear to promising solution concepts for fair division of indivisible objects. They are also compatible with Pareto optimality. The optimal notions can be be further refined with respect to leximin refinements.



There are number of directions for future research. 
	The complexity of finding a weak SD proportional discrete assignment remains open for an unbounded number of agents with non-strict preferences. The complexity of checking whether an SD envy-free assignment exists for a bounded number of agents is open. 
	It will be interesting to see how various approximation algorithms in the literature designed to reduce envy or maximize welfare fare in terms of satisfying ordinal notions of fairness (see, \eg \citep[][]{BeDa05a, BoLa08a, LMMS04a}). Examining other dominance notions may also be fruitful~\citep{GoPe10a}. Another avenue for positive algorithmic results is to consider parameterized algorithms. 
Finally, strategic aspects of ordinal fairness is another interesting direction for future research.

	\section*{Acknowledgments}
	The authors thank the anonymous reviewers of the 13th International Conference on Autonomous Agents and Multiagent Systems (AAMAS 2014) \citep{AGMW14a} as well as the anonymous reviewers of the Artificial Intelligence journal. They also thank Steven Brams for comments and pointers.
	NICTA is funded by the Australian Government through the Department of Communications and the Australian Research Council through the ICT Centre of Excellence Program. Serge Gaspers is the recipient of an Australian Research Council Discovery Early Career Researcher Award (project number DE120101761)
	and a Future Fellowship (project number FT140100048).


\renewcommand\refname{\vskip -1cm}
\section*{References}


  %
  %
  %

\end{document}